\documentclass[sigplan,screen]{acmart}
\usepackage{multirow}
\usepackage{subfig}
\usepackage{soul}
\usepackage{listings}
\AtBeginDocument{%
  }

\newcommand{\stitle}[1]{\vspace{1ex} \noindent{{\bf #1}}}

\newcommand{\name}{\textsc{Bingo}}

\begin{document}

\title{{\name}: Radix-based Bias Factorization for Random Walk on Dynamic Graphs}

\author{Pinhuan Wang}
\affiliation{%
  \institution{Rutgers, The State University of New Jersey}
  \city{Piscataway}
  \state{NJ}
  \country{USA}
}
\email{pw346@connect.rutgers.edu}

\author{Chengying Huan}
\affiliation{%
  \institution{State Key Laboratory for Novel Software Technology, Nanjing University}
  \city{Nanjing}
  \country{China}
  }
\email{huanchengying@nju.edu.cn}

\author{Zhibin Wang}
\affiliation{%
  \institution{State Key Laboratory for Novel Software Technology, Nanjing University}
  \city{Nanjing}
  \country{China}
}
\email{wzbwangzhibin@gmail.com}

\author{Chen Tian}
\affiliation{%
 \institution{State Key Laboratory for Novel Software Technology, Nanjing University}
 \city{Nanjing}
 \country{China}}
   \email{tianchen@nju.edu.cn}

\author{Yuede Ji}
\affiliation{%
  \institution{The University of Texas at Arlington}
  \city{Arlington}
  \state{Texas}
  \country{USA}
  }
  \email{yuede.ji@uta.edu}

\author{Hang Liu}
\affiliation{%
  \institution{Rutgers, The State University of New Jersey}
  \city{Piscataway}
  \state{NJ}
  \country{USA}
 }
\email{hang.liu@rutgers.edu}

\renewcommand{\shortauthors}{Wang et al.}


\begin{abstract} 
Random walks are a primary means for extracting information from large-scale graphs. While most real-world graphs are inherently dynamic, state-of-the-art random walk engines failed to \textit{efficiently} support such a critical use case. This paper takes the initiative to build a general random walk engine for dynamically changing graphs with two key principles: (i) This system should support both low-latency streaming updates and high-throughput batched updates. (ii) This system should achieve fast sampling speed while maintaining acceptable space consumption to support dynamic graph updates. Upholding both standards, we introduce {\name}, a GPU-based random walk engine for dynamically changing graphs. First, we propose a novel radix-based bias factorization algorithm to support constant time sampling complexity while supporting fast streaming updates. Second, we present a group-adaption design to reduce space consumption dramatically. Third, we incorporate GPU-aware designs to support high-throughput batched graph updates on massively parallel platforms. Together, {\name} outperforms existing efforts across various applications, settings, and datasets, achieving up to {a 271.11x} speedup compared to the state-of-the-art efforts.
\end{abstract}




\maketitle

\section{Introduction}

Random walks have drawn increasing attention in recent years since the representation learning of graphs (a.k.a, graph learning) entered the center stage of machine learning~\cite{hamilton2017inductive,nikolentzos2020random,wu2020comprehensive}. Below are two well-known scenarios: (i) In graph learning, a typical option is to use \textit{random walks} to select a few subsets of vertices and edges from the original graph (each of the subsets is treated as a mini-batch) to train the graph neural networks~\cite{tan2023quiver,serafini2021scalable,bajaj2024graph,polisetty2023gsplit,chen2023tango}. This approach increases the scalability, generality, and efficiency of graph learning. However, random walks, unfortunately, take 96.2\% of the end-to-end training time for graph learning, according to Gong et al's paper~\cite{gong2023gsampler}. (ii) In friend recommendation of social media, one uses \textit{random walks} to generate the node embeddings for the final recommendation. The random walk takes 3.5 hours (or 35\% of the total time) on a graph snapshot of 227 million users and 2.71 billion edges~\cite{mei2024flowwalker}. In addition, in personalized PageRank~\cite{haveliwala2003analytical}, SimRank~\cite{jeh2002simrank}, and Random Walk Domination~\cite{li2014random}, we need to launch many \textit{random walks} and use the visit frequency of each vertex across all these random walks as the major indicator to derive PageRank value, vertex similarity, and influence, respectively.


Perhaps the most practical use case of random walks would be extending their capabilities to real-world dynamic graphs. Using fraudulent detection of e-commerce platforms as an example, the transaction graph is changing constantly. The malicious users could commit a series of illicit activities if the graph updates are not immediately integrated into the graph learning process~\cite{qiu2018real}. Therefore, it is imperative to support random walks on dynamic graphs, which is already evident in Ant Finance~\cite{luo2023faf}. Similar needs are also observed in weather forecast~\cite{ma2023histgnn}, Retrieval-Augmented Generation (RAG) of Large Language Models (LLMs)~\cite{lewis2020retrieval}, friend/product recommendation~\cite{yang2021consisrec}, and human resource management~\cite{hang2022outside} among many others~\cite{xia2019random,wang2011understanding}.

Unfortunately, there does not exist, to the best of our knowledge, a {dedicated} system that can \textit{efficiently} support random walks on \textit{dynamically changing} graphs. Below, we discuss related projects in four aspects: 
\ul{(i)} Most existing CPU or GPU-based random walk and graph sampling systems only support random walk on static graphs (KnightKing~\cite{yang2019knightking}, C-SAW~\cite{pandey2020c}, Graphwalker~\cite{wang2020graphwalker}, and others~\cite{jangda2021accelerating,sun2021thunderrw,gong2023gsampler}). Of note, TEA~\cite{huan2023tea} focuses on a temporal graph, a special \textit{static} graph with time-sensitive attributes. 
\ul{(ii)} Existing dynamic graph analytical engines only support traditional graph algorithms, excluding random walks (see most recent LSGraph~\cite{qi2024lsgraph} and the prior ones~\cite{jaiyeoba2019graphtinker,winter2018faimgraph,awad2020dynamic,sha2017technical}).
\ul{(iii)} We do notice two recent projects that support updating random walk based on the graph updates~\cite{papadias2022space,hou2023personalized}. These two projects are orthogonal to {\name}. Particularly, they focus on managing a laundry list of random walks and expediting the process of finding the correct random walk to update. When updating the random walk, they simply rebuild the sampling space for updating, which is inefficient. \ul{(iv)} We also want to clarify that - some recent efforts~\cite{mei2024flowwalker,sun2021thunderrw} called higher-order Random Walk applications, such as Node2vec, as Dynamic Graph Random Walks.  However, these random walk applications are performed on \textit{static graphs} with biases that can change based on the random walk history. For more details on these related works, we refer the readers to Section~\ref{RW}.

This paper aims to design a system that could efficiently support various random walk applications on \textit{dynamically changing graphs}. We identify two important design principles for building such a system:

First, this system should support both low-latency streaming updates and high-throughput batched updates. On the one hand, real-world graphs could experience important updates that should be incorporated immediately. Fraud detection, weather forecast, and RAG of LLM belong to this category. On the other hand, ingesting the system with a batch of graph updates is also commonplace. Certain graph systems, such as product recommendations, could require updating the graph daily with a large volume of updates. As we will demonstrate, batched updates offer more optimization opportunities for a higher ingestion rate.

Second, this system should achieve fast sampling speed while introducing acceptable memory consumption to support dynamic updates. It is anticipated that supporting dynamic graph updates could lead to more complex data structures that will consume more time for sampling and more memory for maintaining the transition probabilities. However, if the penalty on sampling speed and memory consumption is too high, one might simply adopt the sampling on static graphs approach by rebuilding the sampling space from scratch to cope with the dynamic updates. This principle will uphold the quality of this system.

We design and implement {\name}, a GPU-based random walk engine for dynamically changing graphs that adheres to the aforementioned two principles. {Together, {\name} is up to 271.11$\times$ faster than the state-of-the-art that handles graph updates. For graph updates, the ingestion rate of {\name} can reach up to {226} million updates per second.} {\name} features the following four major contributions:

\begin{itemize}
    \item We introduce a novel radix-based bias factorization approach that decomposes the bias of each edge by their radix such that the updates are performed on each bias radix group, which is unbiased. This leads to constant time updating complexity. We also propose a hierarchical sampling algorithm for constant time sampling complexity (Section~\ref{SM}).
    \vspace{.05in}
    \item Considering the memory consumption of our new data structure could be high, we present a group-adaption design to reduce the memory consumption significantly. Since our adaptation is designed according to the nature of various radix groups, this design preserves fast sampling time complexity (Section~\ref{subsec:space}).
    \vspace{.05in}
    \item When handling batched updates, {\name} processes all the requests in a massively parallel manner. Particularly, (i) we introduce a workflow to handle insertions and deletions and potential group rebuilding efficiently. (ii) Our novel delete-and-swap design enables massive parallel deletions (Section~\ref{subsec:bached_updates}).
    \vspace{.05in}
    \item We perform comprehensive evaluations of {\name} on various datasets, configurations, and random walk applications. Overall, {\name} constantly outperforms the state-of-the-art across all graph update settings with acceptable memory consumption. Further, our ingestion rate can reach 0.2 million (streaming updates) and 226 million (batched updates) updates per second across workloads of insertions, deletions, and mixed updates, respectively (Section~\ref{E}).
\end{itemize}

\section{Background}\label{Bg}

This section introduces the definition of dynamic graphs, along with popular random walk applications and the required background about sampling algorithms.

\begin{figure}
    \centerline{\includegraphics[width=1\linewidth]{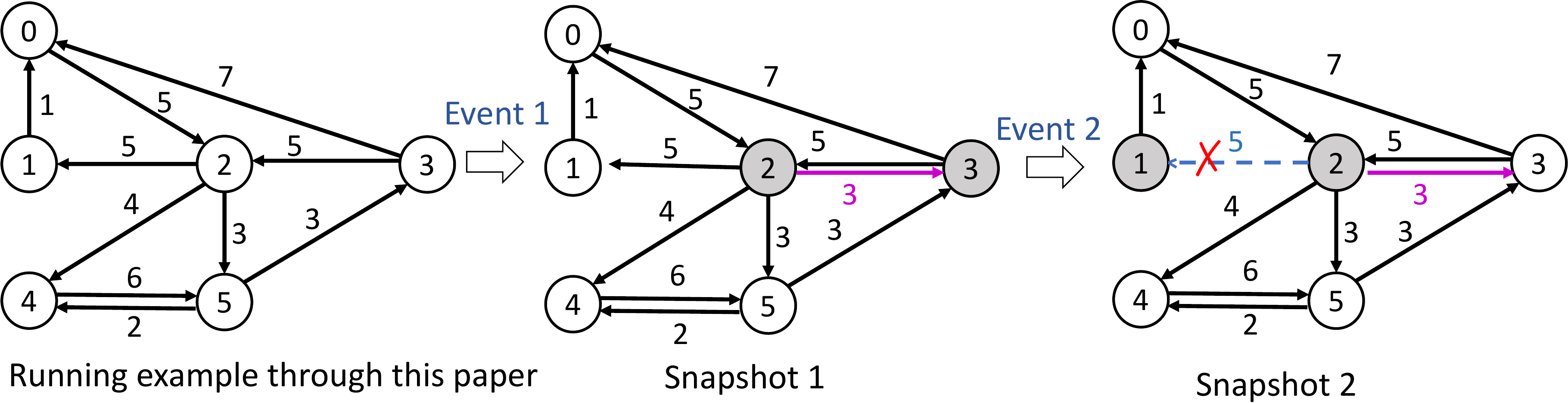}}
    \caption{Running example. Event $1$ contains one edge insertion and event $2$ one edge deletion. 
    \vspace{-.15in}
    }
    \label{dynamic}
\end{figure}

\subsection{Dynamic Graph}

As shown in Figure~\ref{dynamic}, different from static graphs, the states or number of vertices and edges may change over time in dynamic graphs, which is caused by a series of events like vertex update/insertion/deletion or edge update/insertion/deletion. Formally, we define dynamic graphs using the graph snapshot model as follows~\cite{kazemi2020representation}:


\begin{definition}(Dynamic graph).
    A dynamic graph is a sequence of discrete graph snapshots, $\mathcal{G}^t$= \{$\mathcal{V}^t$, $\mathcal{E}^t$\}, where $t \in \mathbb{N}$ represents a timestamp, $\mathcal{V}^t =\{v_1^t, . . .,v_n^t\}$ represents the vertices, and $\mathcal{E}^t = \{e_1^t, . . .,e_m^t\}$ represents the edges.
\end{definition}

Dynamic graphs otherwise behave similarly as static graphs, regardless of whether they have directed or undirected edges, or whether the edges are weighted.

\subsection{Random Walk Applications}

Random walks randomly extract paths from the original large graph. Random walks can be unbiased or biased. Here, we introduce two well-known random walk applications: DeepWalk and node2vec.

\stitle{DeepWalk.} DeepWalk~\cite{perozzi2014deepwalk} is a simple but popular random walk algorithm. It generates random walk paths through repeated sampling. In this process, the walkers stop when they reach the given path length. The paths are treated as sentences and used in the SkipGram model~\cite{mikolov2013efficient} to learn the latent representation. The original DeepWalk was unbiased and extended to a biased version later by Cochez et al~\cite{cochez2017biased}.

\stitle{Node2vec.}\label{sn2v} Different from DeepWalk, node2vec~\cite{grover2016node2vec,zhou2018efficient} is a more expressive higher-order algorithm, in which the transition probability also depends on the walk history. Supposing there is an undirected graph walker at vertex $u$, and the just visited vertex is $w$, the transition probability of an edge to the neighbor $v$ will be multiplied with a factor $f(w,v)$ that depends upon the following conditions:
\begin{equation}
    f(w,v)=
    \begin{cases}
        \frac{1}{p}, & \text{if } distance_{w,v} = 0, \\
        1,           & \text{if } distance_{w,v} = 1, \\
        \frac{1}{q}, & \text{if } distance_{w,v} = 2.
    \end{cases}\label{node2vec}
\end{equation}
Here, $p$ and $q$ are user-defined hyper-parameters that influence the behavior of the walker. Specifically, smaller $p$ increases the tendency for the walker to backtrack, while larger $p$ decreases such a tendency. A lower $q$ encourages the walker to explore far away, like Depth-First Search, while a larger $q$ traps the walker around the group of vertices with strong connections.
Finally, we normalize the transition probability to ensure the sum is 1. 

\subsection{Monte Carlo Sampling Algorithms}\label{sampling methods}

Despite the large volume of random walk applications~\cite{xia2019random}, they all share a common operation - a walker arrives at a vertex and selects among its neighbors for further exploration based on the transition probabilities. This process is called sampling. Sampling is called unbiased when all the candidates share identical probability; otherwise, it is biased sampling. Unbiased sampling is trivial, as we can easily pick the edge through random number generation. For biased sampling, the situation becomes more complex. Each candidate has a bias, which could be the weight or other information associated with the edge or vertex, depending on the specific application. 
Let $w_i$ denote the bias of the $i$-th candidate. Without loss of generality, we consider biased sampling at vertex $u$, which has $d$ neighbors denoted by $\{v_0, \dots, v_{d-1}\}$. The transition probability of selecting neighbor $v_i$ is:
\begin{equation}
    P(v_i) = \frac{w_i}{\sum_{j=0}^{d-1} w_j}.
    \label{transition probability}
\end{equation}

Figure~\ref{sampling} shows the three most common Monte Carlo sampling methods on a static graph. We describe them as follows:

\begin{figure}
    \centerline{\includegraphics[width=1\linewidth]{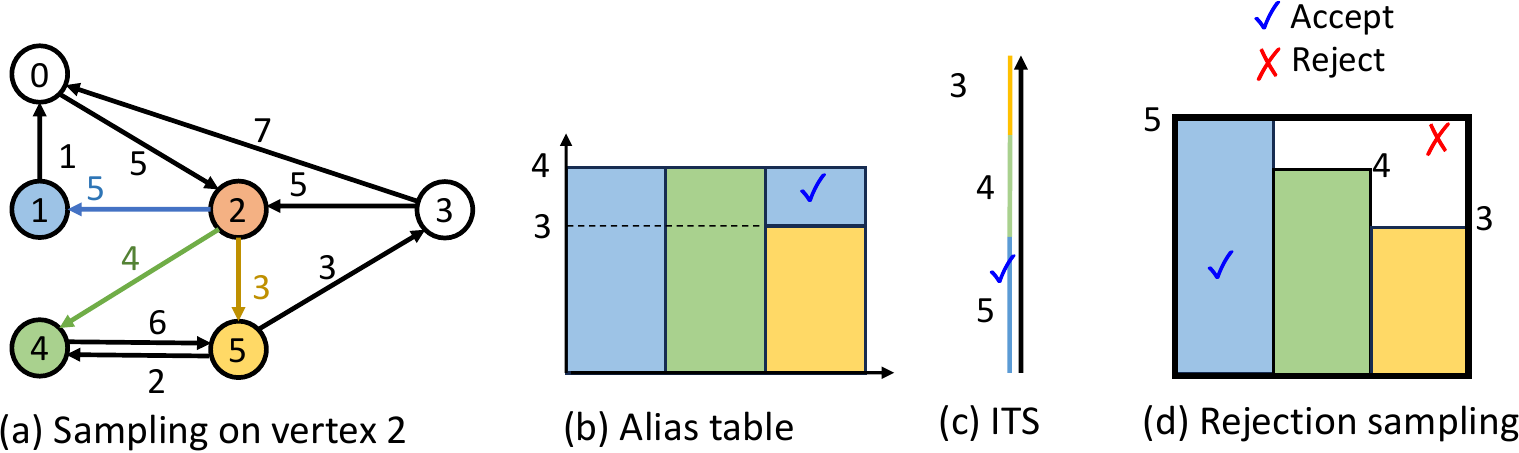}}
    \caption{Three classical Monte Carlo sampling methods for sampling on vertex $2$. 
        }
    \label{sampling}
\end{figure}

\stitle{Alias method} splits all the $d$ ($d$ is the vertex degree) candidates into no more than $2d$ pieces and places them into $d$ buckets, where two specifications are met: (i) Every bucket only contains up to $2$ candidates. (ii) The volume of each bucket should be the same and equal to the average bias of the candidate set. The above structure is the so-called alias table. When sampling, we first select one bucket with equal probability, and then sample among up to $2$ elements in that bucket. The total sampling time complexity is $O(1)$ and the time for alias table construction is $O(d)$. Figure~\ref{sampling}(b) illustrates the alias table-based approach for vertex $2$.

\stitle{Inverse Transform Sampling (ITS)} mainly samples based on an array C, which stores the information of Cumulative Distribution Function (CDF) for candidate transition probabilities. We arrange the candidates into a compact array and construct array $C$ as the prefix sum of their biases, i.e. $c_i=\sum_{j=0}^i{w_i}$ with $c_0=0$. During the sampling process, we randomly generate a number $x$ within the range $[0,c_d)$, then use binary search to determine the interval it belongs to. Specifically, if we find a value $k$ such that $c_{k-1} \leq x < c_k$, then the edge labeled $k$ is the result of sampling. Since we use binary search, the sampling time complexity is $O(\log d)$, while the time complexity for constructing array $C$ is $O(d)$. Figure~\ref{sampling}(c) illustrates an example of the ITS method.

\stitle{Rejection sampling} is commonly used in cases with dynamic bias because it does not require maintaining a table or array like the alias method or ITS. It randomly selects one candidate $i$ with equal probability then decides whether to accept or reject it. We randomly generate a number within the range $[0,max(w))$, where $max(w)$ indicates the maximum bias across all the candidates. 
If this number is less than $w_i$, we accept the candidate; otherwise, we reject it and repeat the above steps to select another candidate. The sampling time depends on the distribution of all the biases and the expected time complexity for rejection sampling is $O(\frac{d\cdot max(w)}{\sum_i w_i})$. In Figure~\ref{sampling}(d), the uncolored area is the rejection area. If we select this area, we need to perform resampling. 

\section{Motivation and {\name} Overview}

\textbf{Motivation.} None of the existing Monte Carlo sampling methods, to the best of our knowledge, is well-suited for sampling on dynamically changing graphs. First, the Alias method~\cite{yang2019knightking,wang2021skywalker,wangl2023optimizing} requires $O(d)$ time complexity for accommodating one edge update, where $d$ is the degree of the affected vertex, although it enjoys $O(1)$ sampling time complexity.
Second, while rejection sampling presents constant time complexity for updating, its rejection rate could be high, see~\cite{yang2019knightking,jangda2021accelerating}.
Third, ITS sampling also enjoys fast updating time but suffers from non-trivial $O(\log d)$ sampling complexity, which is especially true for real-world graphs that reach billions of neighbors~\cite{pandey2020c,gong2023gsampler}.

\textbf{{\name} overview.}
Figure~\ref{workflow} illustrates the major data structures, sampling, and graph updating workflow of {\name}. Particularly, the sampling space is partitioned into various groups. We further build a inter-group sampling space for {\name} to choose the group of interest for sampling. 

\begin{figure}[h]
    \centerline{\includegraphics[width=0.7\linewidth]{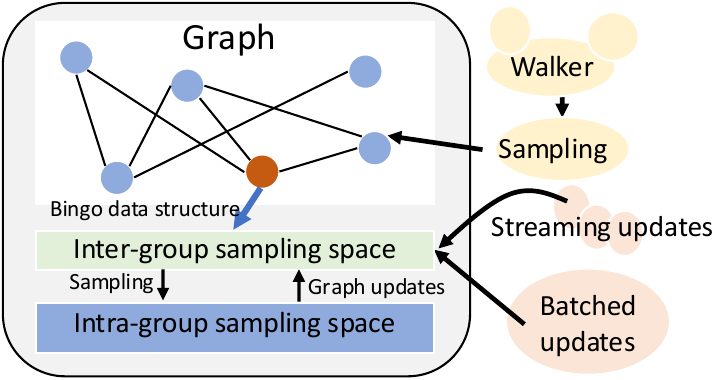}}
    \caption{{\name} workflow. The right top is sampling while the right bottom is updating the graph.
    \vspace{-.1in}
    }
    \label{workflow}
\end{figure}

{\name} features two functionalities: random walk query and graph update. (i) During random walk, when a walker reaches a vertex, it samples from inter-group to intra-group sampling spaces. First, it samples on the per-vertex inter-group sampling space to decide which group to further sample. Second, {\name} moves to that particular group for sampling. This process repeats until the end of the walk.
(ii) 
The updating procedure follows the opposite direction. That is, it begins by deciding which groups will experience updates for a graph update. Subsequently, those groups will be updated. Finally, the per-vertex inter-group sampling space is updated. When batched updates are experienced, we will only update the inter-group space once.

\section{{\name}: Radix-based Bias Decomposition for Random Walk}\label{SM}

\subsection{{\name}: Bias Decomposition Algorithm}

\stitle{High-level intuition:}

When handling vertices with many neighbors, the complexity of rebuilding alias tables or other auxiliary structures increases linearly with the number of neighbors, making traditional sampling techniques computationally expensive. The core idea behind {\name} is to apply a transformation that mitigates this growth by leveraging the binary representation of sampling biases $w_i$. Specifically, by decomposing biases into buckets corresponding to powers of $2$, we effectively reduce the number of values involved in recomputing the alias table from the total number of neighbors to the logarithm of the maximum bias value. This transformation significantly reduces the complexity of incorporating the graph updates. 

\stitle{Sampling space construction:}

We follow a two-step approach for sampling space construction.
(i) We perform a radix-based bias decomposition, decomposing each bias into sub-biases according to the bit positions. Formally, for bias $w_i$, we have a decomposition function $D$ that decomposes $w_i$ into a set of $K$ sub-biases, where $K$ is the number of bits. We term this set as $D(w_i)$. To put it more concretely, we have
\begin{equation}
    \label{eq:dwi}
   D(w_i)=\{2^k|w_i\wedge2^k\neq 0, k=0,1,\dots,K-1\},
\end{equation}
where $\wedge$ represents bitwise AND in this paper. (ii) We reorganize all the $D(w_i)$'s according to the bit positions, $k$, by grouping the sub-biases with the same bit position into the same group. We define the reorganized bias group as $W(p_k)$, where $p_k$ is the $k$-th bit position:

  \vspace{-.05in}
\begin{equation}
    \label{eq:wpk}
    W(p_k)=\sum_{i=0}^{d-1}{w_i\wedge2^k}.
\end{equation}

Notice that if $w_i\wedge2^k=0$, $w_i$ does not contribute to $W(p_k)$, which is equivalent to the fact that $2^k$ is not in $D(w_i)$.

\stitle{Hierarchical sampling.}
{\name} performs a two-stage-based sampling on the aforementioned data structure as follows: 
During stage (i), termed inter-group sampling, we select a group of interest via Monte Carlo sampling. Of note, the bias of a group is the sum of biases in this group. Therefore, the transition probability for group $p_k$ is:
\begin{equation}\label{eq:group_prob}
    P(p_k)=\frac{W(p_k)}{\sum_{j=0}^{K-1} W(p_j)}.
\end{equation}

At stage (ii), we proceed to intra-group sampling. Since group $p_k$ maintains all the neighbors whose bias bit at radix position $k$ is 1, neighbors belonging to the same sub-bias group have equal transition probability.

Thus, we can perform unbiased sampling by randomly picking a vertex $v_i$ from group $p_k$.
The probability is:
\begin{equation}\label{eq:pvp}
    P(v_i|p_k)=\frac{w_i\wedge2^k}{W(p_k)}.
\end{equation}
The corresponding original edge of this sub-bias is the final selected edge.

\begin{theorem}(Correctness). {\name} ensures that the probability of choosing each neighbor remains the same before and after radix-based bias factorization, i.e., Equation (\ref{transition probability}) holds for {\name}'s sampling.
\end{theorem}
\begin{proof}
Although we have transformed the sampling space, the probability of selecting each neighbor remains unchanged because the total bias associated with each neighbor is preserved. Formally, for
    a vertex $v_i$ to be chosen, it can be chosen through all the groups $v_i$ belongs to. Therefore, we arrive at:
    \begin{equation}
        P(v_i)=\sum_{k=0}^{K-1} (P(p_k)\cdot P(v_i|p_k)).
    \end{equation}

    Replacing $P(p_k)$ by Equations~\ref{eq:group_prob} and $P(v_i|p_k)$ by Equation~\ref{eq:pvp}, we arrive at:
    \begin{equation}
        P(v_i)=\sum_{k=0}^{K-1} \frac{w_i\wedge2^k}{\sum_{j=0}^{K-1} W(p_j)}= \frac{\sum_{k=0}^{K-1} w_i\wedge2^k}{\sum_{i=0}^{d-1} w_i}=\frac{w_i}{\sum_{i=0}^{d-1} w_i},
    \end{equation}
which is the same as when we perform sampling before radix-based bias factorization. 
\end{proof}

\begin{figure}[h]
    \centerline{\includegraphics[width=0.95\linewidth]{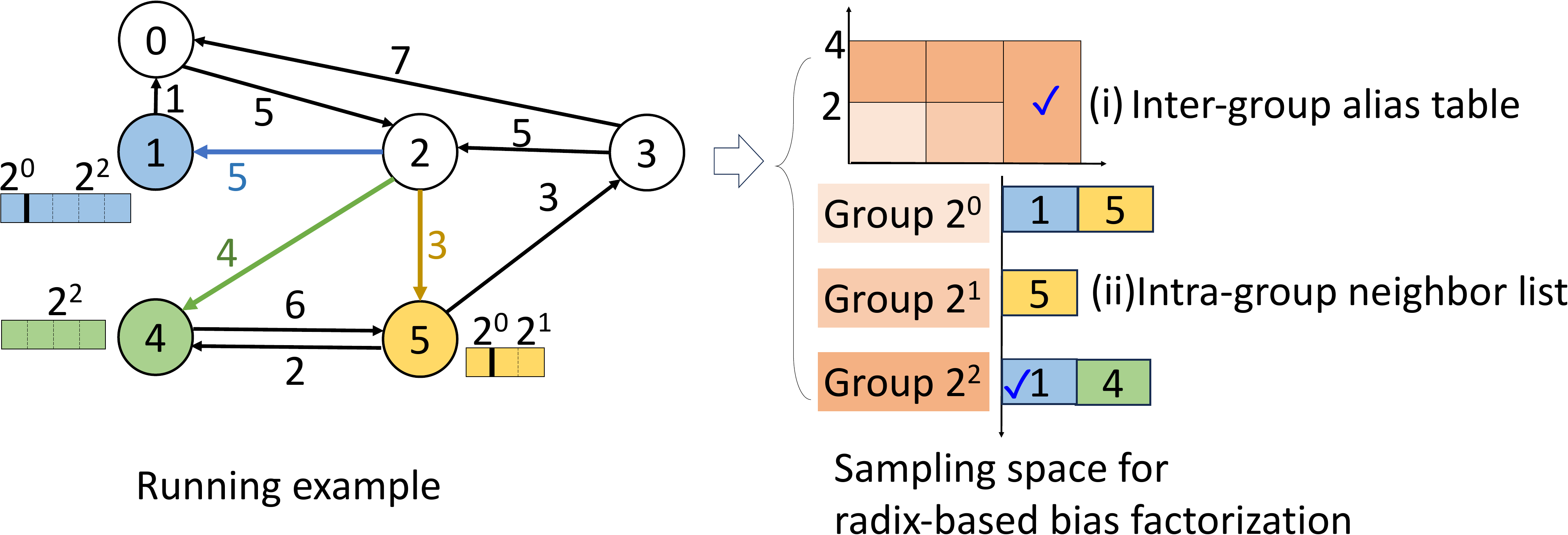}}
    \caption{{\name} on the running example.
    \vspace{-.12in}
    }
    \label{SAMPLING}
\end{figure}

\stitle{{\name} running example.}
Figure~\ref{SAMPLING} exemplifies how to implement {\name} for the running example. Using vertex 2 as an example, following the $(src, dst, bias)$ format,  there are three candidate edges from vertex 2, i.e., $(2,1,5)$, $(2,4,4)$, and $(2,5,3)$. Applying binary factorization to the biases of all three neighbors, we get group $2^0$ contains $\{1,5\}$, group $2^1$ contains $\{5\}$, and group $2^2$ contains $\{1,4\}$. Here, we directly use the out-neighbor ID to represent the edge. Therefore, the biases of these three groups are 2, 2, and 8. We adopt the alias table method to build the (i) inter-group sampling space, which is shown in the top right of Figure~\ref{SAMPLING}.

During sampling, step (i) selects group $2^2$ in the inter-group alias table. Subsequently, step (ii) relies on unbiased sampling to select neighbor 1 in group $2^2$.

\stitle{Remarks on {\name} sampling benefits.} {\name} enjoys fast sampling speed. Particularly, in stage (i), we adopt the alias method in stage(i), which offers constant time complexity. Further, because the number of groups is small, 
updating the alias table will be fast. For stage (ii), while the sampling space is large, it is uniform sampling. Therefore, the sampling speed is, again, fast (i.e., constant time complexity).

\subsection{{\name} for Streaming Updates}\label{streaming}

Now, we introduce our designs to achieve fast graph updates (i.e., insertion and deletion) in {\name}.

\begin{figure}[h]
    \centerline{\includegraphics[width=0.85\linewidth]{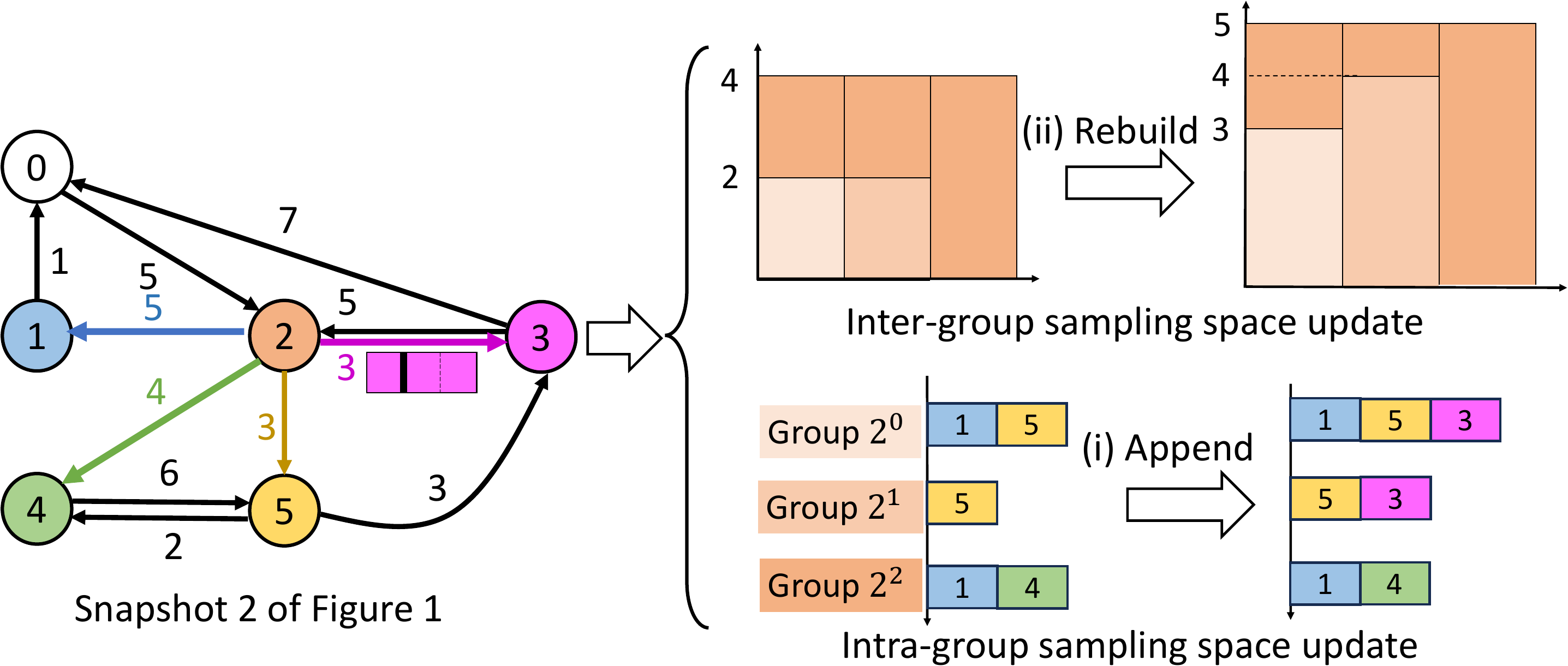}}
    \caption{{\name} insertion operation.}
    \label{insertion}
     \vspace{-.1in}
\end{figure}

\stitle{\textbf{Insertion:}} Figure~\ref{insertion} exemplifies the insertion operation of {\name}. We first split the newly inserted edge following Equation~\ref{eq:dwi} and reorganize these sub-biases following Equation~\ref{eq:wpk}.
In this example, the new edge (2, 3, 3)'s bias is decomposed into $3=2^0+2^1$. We thus split this edge into two sub-biases, one belonging to group $2^0$, and the other to $2^1$. 

During insertion, we perform the intra-group updates before the inter-group one: First, one can simply (i) append the new edge to the end of each neighbor list array. Second, the inter-group alias table is updated based on the new biases in all groups, which is shown as (ii) rebuild in Figure~\ref{insertion}.

\begin{figure}[h]
    \centerline{\includegraphics[width=1\linewidth]{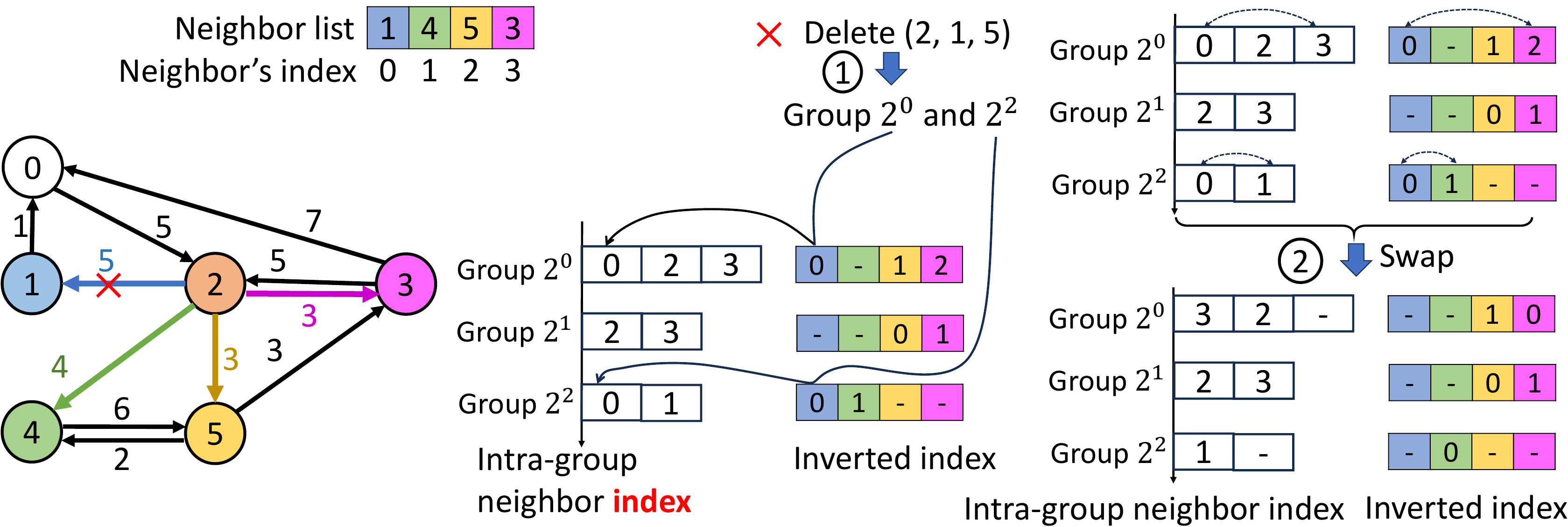}}
    \caption{{\name} deletion operation.
    \vspace{-.17in}
    }
    \label{deletion}
\end{figure}

\stitle{Deletion:} Figure~\ref{deletion} illustrates the edge deletion operation of {\name}. This process contains four steps: (i) we perform a radix-based bias decomposition to identify which groups this edge has contributed sub-biases to.
(ii) One needs to locate this edge in those identified groups. (iii) We swap this edge with the edge in the tail of that group to maintain a compact intra-group neighbor list for $O(1)$ unbiased sampling. (iv) We rebuild the inter-group sampling space similar to insertion in Figure~\ref{insertion} thus omitted for brevity.

We introduce two design changes to achieve a near-constant time complexity for deletion. First, we store the neighbor index in each group as opposed to the neighbor ID, which was the case in Figures~\ref{SAMPLING} and~\ref{insertion}, because the neighbor index can locate the edge in the neighbor list in O(1) complexity. 
This will help the swap operation in step (iii).
Second, we introduce an inverted index that tracks where each neighbor index was stored in each group. {Formally, this inverted index maintains a mapping from the neighbor index to the position in each group.} This reduces the complexity of locating edges in step (ii) from $O(d)$ to $O(1)$.

Figure~\ref{deletion} exemplifies our deletion operation. Assuming we want to delete edge (2,1,5). In step (i), we identify this edge contributes to groups $2^0$ and $2^2$. (ii) Because the edge (2,1,5)'s index is 0, we obtain the location of this edge in these two groups at the first entry of the corresponding inverted index, both of which are 0 in the inverted index. (iii) We swap them with the tail of that group. Using group $2^0$ as an example, the neighbor index 0 will swap with neighbor index 3. Because we store the neighbor index in the intra-group neighbor index, we use this index to locate that we should update the pink location to 0. In a nutshell, one can use the content in the intra-group neighbor index to locate where we should change in the inverted index and vice versa.

In addition to insert/delete edges, other graph updates, e.g., deleting a vertex, and updating the edge bias, can be either implemented with insertion and/or deletion operations or supported straightforwardly. We omitted their descriptions for the sake of space constraints.

 \vspace{-.05in}
\subsection{{\name} for Floating-Point Biases}

While one might suggest that radix sort~\cite{merrill2010revisiting} is the closest related work to {\name} when handling floating-point biases, {\name}, faces more stringent requirements. That is, radix sort only cares about whether a bias is bigger or smaller than the other bias to sort out the order while {\name} requires to know one bias is exactly how much bigger or smaller than the other to build the sampling space.
For more details about radix sort of floating-point values, we refer the readers to~\cite{radixSortRevisited}.
In that regard, {\name} cannot adopt radix sort ideas to handle floating-point biases in sampling.

\stitle{\name's approach.}
{\name} handles floating-point biases in a four-step approach. (i) We empirically determine an amortization factor \( \lambda \), which is used to round floating-point values to proportional integer values.  (ii) For each bias, we decompose this bias into an integer part and a decimal part. (iii) We further perform radix decomposition for the integer part of all the biases (similar to our integer alone case) while leaving the decimal part in one group. If the decimal part is taken, we will adopt ITS or rejection sampling. (iv) We perform the hierarchical sampling to choose the edge of interest. Of note, $\lambda$ is properly chosen such that the sum of the decimal parts is very low. Then the sampling time complexity remains low. See Section~\ref{subsec:complex} for detailed discussion.

\begin{figure}[h]
    \centerline{\includegraphics[width=0.95\linewidth]{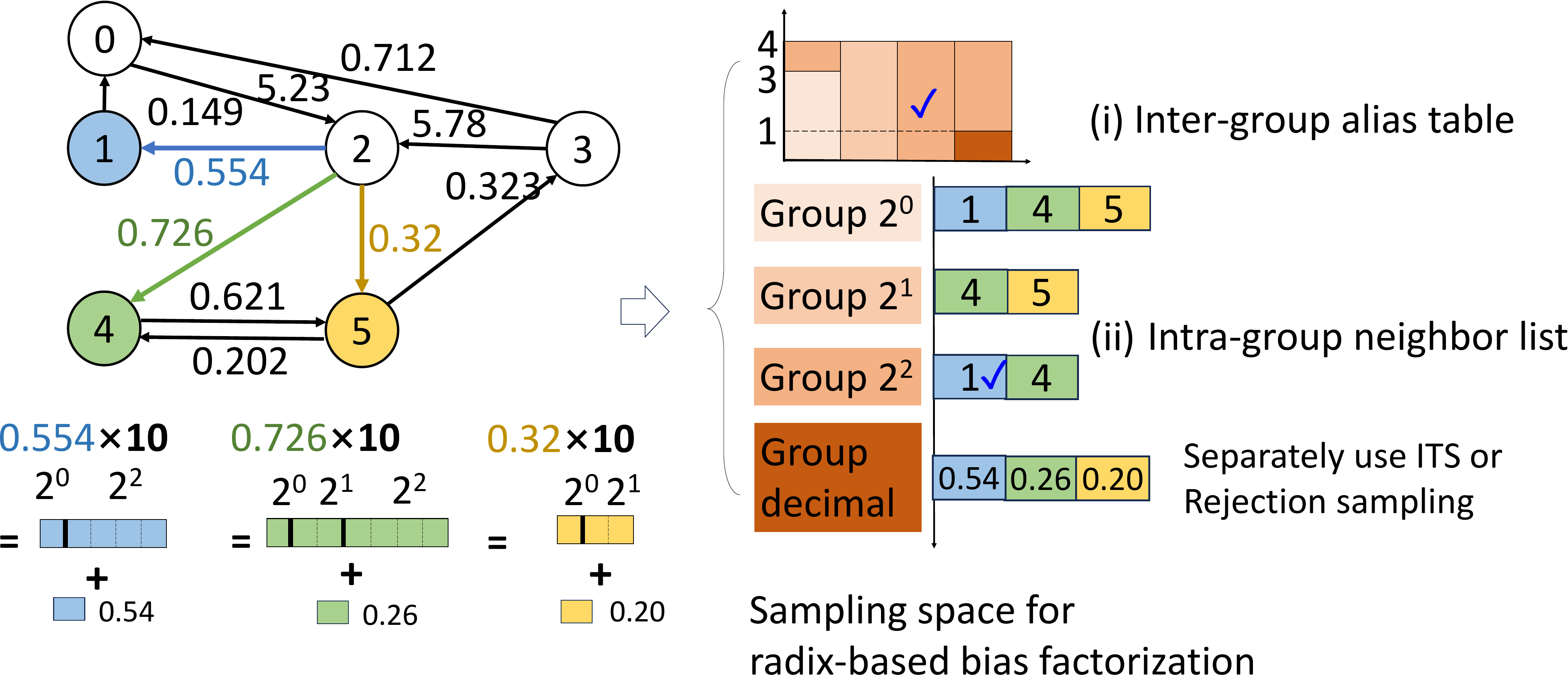}}
    \caption{{\name} with floating-point biases.}
    \label{SAMPLING2}
\end{figure}

Figure~\ref{SAMPLING2} illustrates how {\name} works on the floating-point biases. Still using vertex 2 as an example, there are three candidate edges from vertex 2, i.e., $(2,1,0.554)$, $(2,4,0.726)$, and $(2,5,0.320)$. We multiply all the biases by $10$. Then we get edges with new biases, $(2,1,5.54)$, $(2,4,7.26)$, and $(2,5,3.20)$.
Applying binary factorization to the integer part of biases of all three edges, we get groups $2^0$, $2^1$ and $2^2$, respectively, containing neighbors $\{1,4,5\}$, $\{4,5\}$, and $\{1,4\}$. Besides, we collect the decimal part as a new group, i.e., \{1, 4, 5\}. We adopt the alias table method to build the (i) inter-group sampling space, which is shown in the top right of Figure~\ref{SAMPLING2}. Since the decimal group only takes up a small area of the alias table, for most of the cases, our sampling happens in the integer part, whose sampling complexity is constant.

\subsection{Complexity Analysis}\label{subsec:complex}

Table \ref{timecmp} presents the complexity of {\name} vs four common Monte Carlo sampling methods. 
The superiority of {\name}'s time complexity can be clearly observed from this table. However, {\name} consumes more memory than others. We will introduce a series of optimizations for memory consumption in the following section.

\stitle{Time complexity}. 
For a vertex with a degree of $d$, we assume each neighbor has the bias of $w_i$. (i) Sampling:
Given that both inter- and intra-group sampling (alias table sampling and unbiased sampling) have a time complexity of $O(1)$, the total sampling time remains $O(1)$. 
(ii) Insertion: For each of the $K$ groups, we preform $O(1)$ insertions. The size of $K$ could be derived by $K=log(max(w_i))$. Additionally, we need to reconstruct the alias tables for each of these groups. Therefore, the overall time complexity of insertion is $O(K)$.
(iii) Deletion: as deletion operations within each group are identical, each single deletion step in each group only consumes $O(1)$ constant time, and the time complexity of deletion is also $O(K)$. 

{For floating-point biases, we only need to analyze the sampling of the intra-group part (inter-group sampling complexity is O(1)). Let $W_I$ and $W_D$ denote the sums of the biases for, respectively, integer and decimal parts. The sampling time complexity becomes $O(\frac{W_D}{W_I+W_D}\cdot \frac{d\cdot max(w_i)}{\sum_i w_i}+\frac{W_I}{W_I+W_D})$, where use rejection sampling as an example. In this expression, the first term means inter-group sampling selects the decimal group, while the second term means the integer group. Through adjusting $\lambda$, we can ensure that $\frac{W_D}{W_I+W_D}<\frac{1}{d}$, which keeps our sampling complexity as $O(1)$. In Figure~\ref{SAMPLING2}, we set $\lambda=10$ which leads $\frac{W_D}{W_I+W_D}=\frac{1}{16}<\frac{1}{d}=\frac{1}{3}$. This $\lambda$ ensures $O(1)$ sampling complexity.

The updating time complexity of {\name} is only related to the number of groups, i.e., $K$, which is small (usually no more than $32$ or $64$), which is the number of bits for integer and long integer.A larger radix base can further reduce $K$, which is briefly discussed in Section~\ref{sec:disc} (see supplement material).
\begin{table}
    \caption{{Complexity comparison for a vertex: {\name} vs. Alias, ITS, and Rejection sampling, where $K$ and $d$ are the numbers of groups and degree of the vertex.}}
    \vspace{-.1in}
    \begin{center}
        {\scriptsize
        \begin{tabular}{|c|c|c|c|c|}
            \hline
            \textbf{Name} & \textbf{Insertion}      & \textbf{Deletion}       & \textbf{Sampling}                                         & \textbf{Memory}         \\
            \hline
            {\name}       & $O(K)$                  & $O(K)$                  & O(1)                                                      & $O(d\cdot K)$           \\
            \hline
            Alias         & \multirow{2}{*}{$O(d)$} & \multirow{2}{*}{$O(d)$} & \multirow{2}{*}{$O(1)$}                                   & \multirow{2}{*}{$O(d)$} \\
            Method        &                         &                         &                                                           &                         \\
            \hline
            ITS           & $O(1)$                  & $O(d)$                  & $O(log_2d)$                                               & $O(d)$                  \\
            \hline
            Rejection     & \multirow{2}{*}{$O(1)$} & \multirow{2}{*}{$O(d)$} & \multirow{2}{*}{$O(\frac{D\cdot{max(w)}}{\sum_i w_i})$} & \multirow{2}{*}{$O(d)$} \\
            Sampling      &                         &                         &                                                           &                         \\
            \hline

\end{tabular}
}
        \label{timecmp}
    \end{center}
    \vspace{-.2in}
\end{table}

\stitle{Space complexity.}
Our naive design repeats the inverted index $K$ times, i.e., the number of groups, with the size of each inverted index as $d$. This leads to a memory consumption of $d\cdot K$. Further, each edge appears in $t=popc(w_i)$ groups, i.e., the number of nonzero bits in bias $w_i$. This leads to a memory consumption of $d\cdot t$ for the intra-group neighbor index list. Combined, our method amplifies the memory consumption from $d$ to $d\cdot(K+t)$. Since each vertex might not appear in all groups, we derive $t\leq K$. Therefore, the space complexity of {\name} is $O(d\cdot K)$.

\section{{\name} System Implementation and Optimizations}
\label{sec:opt}

This section implements {\name} on GPUs with two optimizations, i.e.,  memory consumption optimization (Section~\ref{subsec:space}) and batched graph updates optimization (Section~\ref{subsec:bached_updates}). 
\vspace{-.1in}
\subsection{Adaptive Group Representation}
\label{subsec:space}

\begin{figure}[h]
    \vspace{-.2in}
    \centerline{\includegraphics[width=1\linewidth]{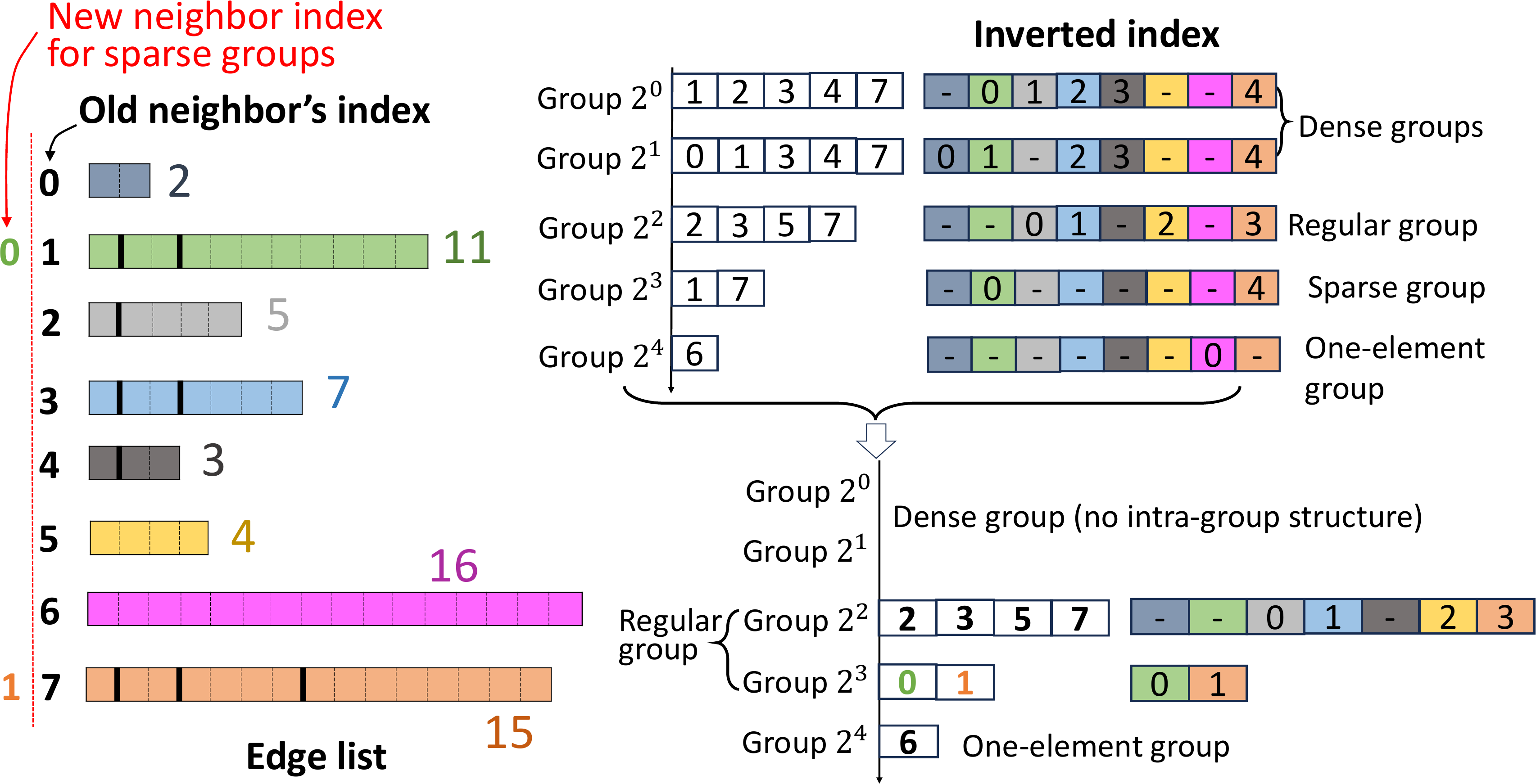}}
    \vspace{-.1in}
    \caption{{\name} adaptive group representation.
    }
    \label{adaptation}
\end{figure}

\noindent
Figure~\ref{adaptation} presents our solutions to reduce both $K + t$ and $d$ (see Section~\ref{subsec:complex}), which reduces the memory consumption for {\name}. Our intuition is as follows: we first separate dense and sparse groups from the regular group, which permits simpler data structures (i.e., less memory consumption). Additionally, we observe that one-element groups constitute a significant proportion of the data (see Figure~\ref{Moptimization}), and handling them separately enables further memory savings through simplified representation. Based on the cardinality of different groups $G_i$, we divide them into four categories: dense, one-element, sparse, and regular groups, following Equation~\ref{eq:group}:

\begin{equation}\label{eq:group}
{\small
    \text{Group\ $G_i$} \in
    \begin{cases}
        \text{Dense group,}       & \text{$\frac{|G_i|}{d}$ >$\alpha$\%;}                    \\
        \text{One-element group,} & \text{|$G_i$| = 1;}                                         \\
        \text{Sparse group,}      & \text{$\frac{|G_i|}{d}$ < $\beta$\% and $|G_i| \neq 1$;} \\
        \text{Regular group,}     & \text{Otherwise.}
    \end{cases}
    }
\end{equation}

Based on our heuristic study, we set $\alpha=40$ and $\beta=10$ in our design for the optimal performance.

\stitle{Dense group} is often the groups with less significance. Using group $2^0$, i.e., the group with the least significant bit as an example, a neighbor $v_i$ falls into this group when $w_i\mod2=1$, or simply with an odd bias. Statistically, half of the neighbors (50\%) has a chance of falling into the $2^0$ group.

We propose a radical change for dense groups, i.e., we maintain neither the intra-group neighbor index list nor the inverted index for a dense group. This helps reduce the memory cost for both $t$ and $K$ (mentioned above). Besides, these groups are not accessed frequently because they often appear in the groups with less significant bits. Therefore, we save the memory for the largest intra-group neighbor index lists and the inverted indices with potentially small or no sampling efficiency impacts.

We adopt a rejection sampling on the original neighbor list for dense groups. Even if dense groups are selected, we can still maintain a low rejection ratio because the rejection ratio is below (1-$\alpha$\%)=60\%. Particularly, our rejection sampling works as follows: (i) Our inter-group sampling selects a particular dense group. (ii) For intra-group sampling, we select a neighbor from the original neighbor list. Further, we use that {neighbor's bias} to AND (i.e., \&) the group radix of the chosen dense group. If the result is not zero, this sampling is accepted. Otherwise, we repeat the sampling from the intra-group sampling.

In Fig~\ref{adaptation}, group $2^0$ and group $2^1$ are dense groups, since they contained 62.5\% of edges. If inter-group sampling selects group $2^0$, we will further randomly select one neighbor from the original neighbor list. Assuming we selected the 5-th neighbor, we use its bias 4 to AND $2^0$, which returns 0. This means our sampling is rejected. We will repeat this process until we find a neighbor in group $2^0$.

\stitle{One-element group.}
With skewed bias distributions, chances are group(s) with the most significant values might only contain one element. For those groups, there is no need to maintain either the inverted index or the intra-group neighbor index list. Group} $2^4$ in Figure~\ref{adaptation} belongs to this case.

\begin{figure}[t!]
    \centerline{\includegraphics[width=.8\linewidth]{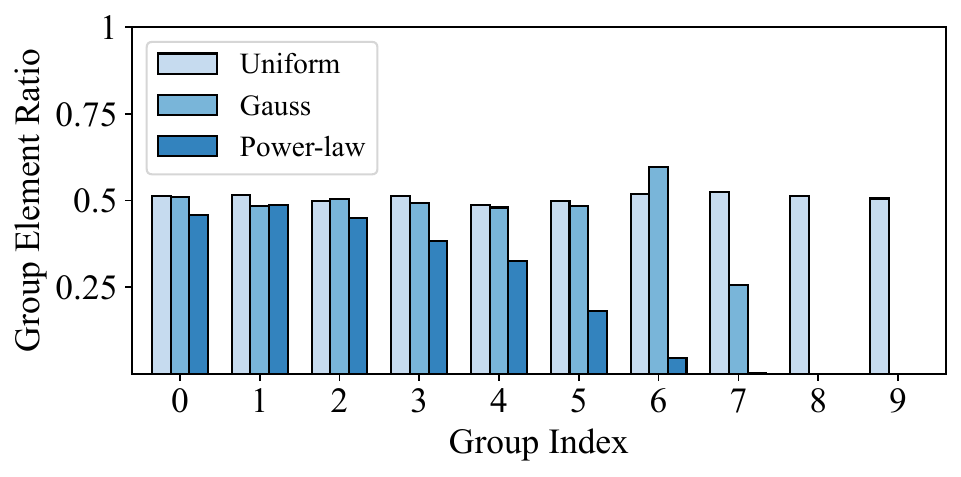}}
    \vspace{-.2in}
    \caption{Group element ratio of various distributions.
    \vspace{-.2in}
    }
    \label{distribution}
\end{figure}

\stitle{Sparse group.}
Our sparse group optimizations mainly focus on reducing the memory consumption of the inverted index.
{As shown in Figure~\ref{distribution}, for two out of three cases, groups with higher powers tend to have fewer edges, which means the inverted indexes in these groups can be very sparse.} 
Consequently, we maintain a new neighbor list that only contains edges with larger biases. This leads to a much smaller neighbor list, thus a smaller inverted index for sparse groups.

As shown in Figure~\ref{adaptation}, neighbor indices 1, 6, and 7 are the edges with bias larger than 8. Since neighbor 7 also belongs to the one-element group, we exclude it from the sparse group. Therefore, our new neighbor list contains \{1, 7\}, whose cardinality is 2. In this case, our inverted index has a size of 2 for group $2^3$, which is significantly smaller than the original inverted index with a size of 8.

\stitle{Regular group.} After filtering out dense, one-element, and sparse groups, we might still have a few groups. These groups require the full inverted index and the original intra-group neighbor index list. In Figure~\ref{adaptation}, group $2^2$ is such a group.

\subsection{{\name}'s Parallel Batched Graph Updates}
\label{subsec:bached_updates}

Figure~\ref{fig:update}(a) presents the overall workflow of our parallel batched graph updates. On the CPU platform, we first put the graph updates of the same vertex together. Subsequently, we move these ordered update requests to GPU. For each vertex $v_i$, we, in order, perform three steps, i.e., insert, delete, and rebuild. Of note, one might insert a just deleted edge back; we thus allow duplicated insertions of the same edge with a time stamp. When deletion happens to a duplicated edge, we delete the earlier version first.

During insertion, we first insert each edge into the graph following the dynamic graph format~\cite{busato2018hornet}. Subsequently, we perform four types of insertions differently. \ul{For dense group}, we simply do nothing because it does not maintain any sampling-related data structures. \ul{One-element group}: We derive whether this group evolves into a sparse/regular/dense group based on all the insertions. Subsequently, it follows how an insertion is done to a sparse/regular/dense group to perform this particular insertion. \ul{Sparse group:} We first append this new neighbor to the sparse group neighbor list. Second, we append to the intra-group neighbor index list and the inverted index of this group for this specific neighbor. \ul{Regular group}: we append the edge to the intra-group neighbor index list and update the inverted index.

\begin{figure}
    {\includegraphics[width=0.99\linewidth]{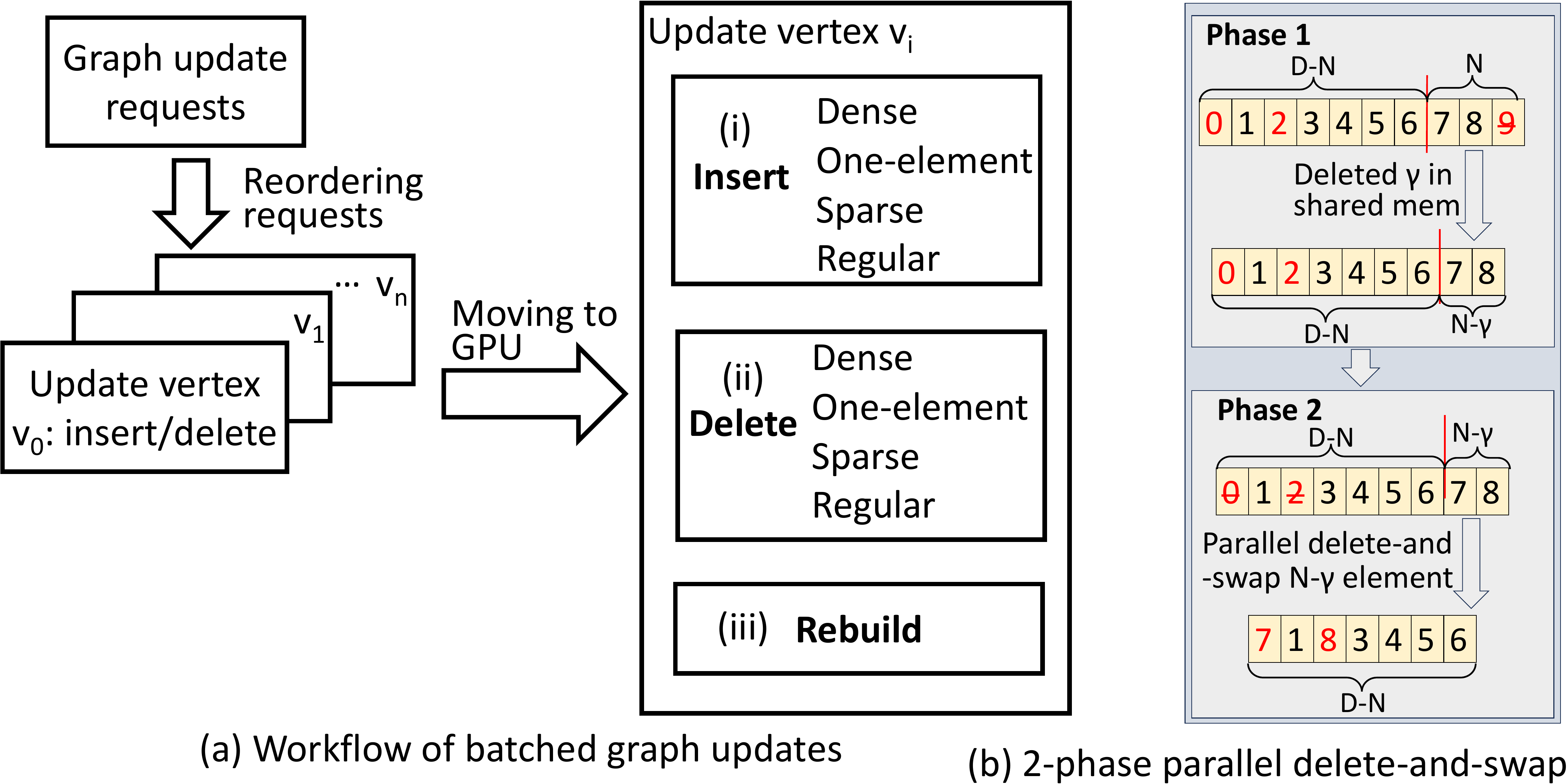}}
    \caption{Parallel batched graph updates.
    \vspace{-.2in}
    }
    \label{fig:update}
\end{figure}

Deletion is more complex than insertion. \ul{On-element group:} We simply remove this group. \ul{Dense group:} we will perform a rebuild if this dense group evolves into a different group type (details discussed later).
\ul{Sparse/Regular group:} Since parallel delete-and-swap could introduce bubbles in the intra-group neighbor index list that prevent unbiased sampling, we introduce a design to support parallel delete-and-swap as below:

Figure~\ref{fig:update}(b) exhibits our two-phase parallel delete-and-swap N elements operation, which is a key contribution to our parallel deletion. The biggest concern of the parallel delete-and-swap is that the elements from the tail might also be deleted. Therefore, if one uses the concurrently deleted tail element to fill in the entry that is to be deleted, one fills the to-be-deleted entry with a voided entry. For instance, one wants to delete entry 0 and use entry 9 to fill in entry 0. However, what if entry 9 will also be deleted? In this case, we cannot use entry 9 to fill entry 0.

In our design, we perform this deletion in two phases. (i) We load N elements from the tail to GPU's shared memory, if fit. Otherwise, we keep them in global memory. Subsequently, we will delete all the elements that are supposed to be deleted in these N elements. We assume we have deleted $\gamma$ elements. (ii) Since we have already deleted $\gamma$ elements from the list, we only need to delete $N-\gamma$ more elements from the list. In phase (i), we are sure the remaining $N-\gamma$ elements will not be deleted. Therefore, we are guaranteed to delete $N-\gamma$ elements at the front with these $N-\gamma$ will-not-deleted elements to fill the entry.

Our rebuild mainly tackles group-type transformations of the following two cases: (i) from other group formats to a regular group, and (ii) From other group formats to a sparse group format. The reason for case (i) is one needs to rebuild the entire intra-group neighbor index list and inverted index, which is expensive during insertion or deletion. For case (ii), one needs to scan the original group to rebuild the new neighbor index, which is again better performed after all the insertions and deletions are accomplished.

\section{Evaluation}\label{E}

{\name} is implemented using CUDA/C++ with approximately 2,000 lines of code. It is designed with an efficient GPU-based architecture to handle graph processing and random walks. Below are the implementation details: (i) {\name} stores the graph and related metadata directly on the GPU, ensuring fast access and efficient computations. (ii) Before each random walk computation, Bingo integrates all the graph updates, maintaining the correct order of operations to ensure consistency. This correctness is enforced for both streaming and batched random walks. (iii) Bingo performs random walks in a step-by-step manner, where each step involves sampling to select the next node for the next step of random walks. (iv) Bingo treats each vertex as an individual object and utilizes eight main kernels to support various random walk applications, i.e., streaming\_insert, streaming\_delete, batched\_insert, batched\_delete, random\_walk\_node2vec, random\_walk\_deepwalk, random\_walk\_ppr, and random\_walk \_simple\_sampling.

\subsection{Experimental Setup}

The evaluation is performed on a Linux server equipped with two $2.8 GHz$ Intel(R) Xeon(R) Silver $4309$Y CPUs, $16$ physical cores, $512$ GB of memory, and four NVIDIA A100 GPUs, each with $80$ GB of HBM2e memory.

\stitle{Graph Applications.}
We examine {\name} and the state-of-the-art with three applications: biased DeepWalk, node2vec, Personalized PageRank (PPR). 
For all of them, we initialize the vertex count number of random walkers.
Further, for DeepWalk and node2vec, we set the default walk length as $80$. We set the hyper-parameters $p=0.5$ and $q=2$ for node2vec. The parameters for DeepWalk and node2vec are identical to that from KnightKing~\cite{yang2019knightking}. We put the termination probability of PPR as $1/80$, which offers an expected walk length of $80$. 

\begin{table}[h]
    \caption{Graph dataset ($K=10^3$, $M=10^6$).
            \vspace{-.1in}
}
    \begin{center}
    {\scriptsize
        \begin{tabular}{|c|c|c|c|c|c|}
            \hline
            \multirow{2}{*}{\textbf{Dataset} } & \multirow{2}{*}{\textbf{Abbr.} } & {\textbf{Vertex}} & \textbf{Edge}  & \textbf{Avg}    & \textbf{Max}    \\
                                               
            &                & \textbf{count}    & \textbf{count} & \textbf{degree} & \textbf{degree} \\
                                               
            \hline
            Amazon  & AM   & 403.4K  & 3.4M    & 8.4  & 10              \\
            \hline
            
            Google    & GO & 875.7K   & 5.1M   & 5.8  & 456             \\
            
            \hline
            Citation  &CT&3.8M&16.5M&4.4&770\\
            
            \hline
            LiveJournal                        & LJ                               & 4.8M              & 68.5M          & 14.3            & 20.3K           \\
            \hline
            
            Twitter     & TW         & 41.7M             & 1,468.4M       & 35.2            & 770.2K          \\
            \hline
        \end{tabular}
   }
        \label{dataset}
    \end{center}
\end{table}

\begin{table*}[]
\caption{{\name} vs. SOTA on time and memory consumption. }
 \vspace{-.05in}
\label{tab:overall}
{\normalsize
\begin{tabular}{c|c|c|c|c|c|c|c}
\hline
\multirow{2}{*}{\textbf{Algorithms}} & \multirow{2}{*}{\textbf{Frameworks}} & \multicolumn{5}{c|}{\textbf{Runtime (s),  Memory Consumption (GB)}} & \textbf{Avg.} \\ \cline{3-7}
 &  & \multicolumn{1}{c|}{Amazon} & \multicolumn{1}{c|}{Google} &Citation& \multicolumn{1}{c|}{LiveJournal} &  Twitter &  \textbf{speedup}\\ \hline
 
 & {\name} & \multicolumn{1}{c|}{\colorbox{green}{0.51}, 0.70} & \multicolumn{1}{c|}{\colorbox{green}{1.01}, 1.54} &\colorbox{green}{0.62}, 2.09& \multicolumn{1}{c|}{\colorbox{green}{1.99}, 6.34}  & \colorbox{green}{19.63}, 51.22 & -\\\cline{2-8} 
\textbf{DeepWalk}& KnightKing & \multicolumn{1}{c|}{2.64, 0.17} & \multicolumn{1}{c|}{4.39, 0.17} &12.63, 0.80& \multicolumn{1}{c|}{70.92, 2.54} &  1114.26, 46.14 & 24.46 \\ \cline{2-8} 

 \textbf{Insertion}& gSampler & \multicolumn{1}{c|}{1.55, \colorbox{pink}{1.25}} & \multicolumn{1}{c|}{2.44, \colorbox{pink}{1.99}} & 3.81, \colorbox{pink}{7.95} &\multicolumn{1}{c|}{22.97, \colorbox{pink}{8.80}} & 403.91, \colorbox{pink}{71.24} & 8.74 \\ \cline{2-8}
 &FlowWalker&0.55, 0.70&1.85, 1.02&4.26, 2.99&11.01, 4.04&25125.48, 43.97 &259.05
 \\ \hline

 & {\name} & \multicolumn{1}{c|}{\colorbox{green}{0.25}, 0.70} & \multicolumn{1}{c|}{\colorbox{green}{0.32}, 1.44} &\colorbox{green}{0.16}, 2.09& \multicolumn{1}{c|}{\colorbox{green}{1.46}, \colorbox{pink}{6.24}} &  \colorbox{green}{19.11}, 51.12 & -\\ \cline{2-8} 
\textbf{DeepWalk}& KnightKing & \multicolumn{1}{c|}{2.48, 0.17} & \multicolumn{1}{c|}{4.43, 0.17} &12.41, 0.78& \multicolumn{1}{c|}{71.45, 2.54}  & 1136.11, 46.14 & 41.94\\ \cline{2-8} 
 \textbf{Deletion}& gSampler & \multicolumn{1}{c|}{1.47, \colorbox{pink}{1.01}} & \multicolumn{1}{c|}{2.33, \colorbox{pink}{1.51}} &
 3.44, \colorbox{pink}{5.54}
 &\multicolumn{1}{c|}{21.46, 5.92} & 376.54, \colorbox{pink}{60.61} & 13.81 \\\cline{2-8}
 &FlowWalker&0.56, 0.70&1.84, 1.02&4.20, 2.99&10.44, 4.04&25112.71, 43.97&271.11
 \\\hline

 & {\name} & \multicolumn{1}{c|}{\colorbox{green}{0.25}, 0.87} & \multicolumn{1}{c|}{\colorbox{green}{0.64}, 1.38} &\colorbox{green}{0.44}, 4.62& \multicolumn{1}{c|}{\colorbox{green}{1.96}, 7.25} &  \colorbox{green}{18.94},  51.22 &-\\\cline{2-8}

\textbf{DeepWalk}& KnightKing & \multicolumn{1}{c|}{2.59, 0.17} & \multicolumn{1}{c|}{4.52, 0.22} &12.40, 0.78& \multicolumn{1}{c|}{67.43, 2.53}  & 1006.66, 46.21 & 26.63\\ \cline{2-8} 
 \textbf{Mixed}& gSampler & \multicolumn{1}{c|}{3.30, \colorbox{pink}{1.25}} & \multicolumn{1}{c|}{2.64, \colorbox{pink}{2.14}} &
 4.39, \colorbox{pink}{7.72}
 &\multicolumn{1}{c|}{ 17.60, \colorbox{pink}{10.08}} & 307.16, \colorbox{pink}{73.19} & 10.50 \\\cline{2-8}
 &FlowWalker&0.57, 0.70&1.94, 1.02&4.33, 2.99&10.72, 4.04&25137.66, 43.97&269.57
 
 \\\hline

 & {\name} & \multicolumn{1}{c|}{ \colorbox{green}{0.63}, 0.79} & \multicolumn{1}{c|}{ \colorbox{green}{1.18}, 1.64} & \colorbox{green}{0.66}, 2.29&\multicolumn{1}{c|}{ \colorbox{green}{6.35}, 7.03} &  \colorbox{green}{66.75}, \colorbox{pink}{66.55} & -\\\cline{2-8}
\textbf{node2vec}& KnightKing & \multicolumn{1}{c|}{13.75, 0.33} & \multicolumn{1}{c|}{15.47, 0.44} &58.23, 1.00&\multicolumn{1}{c|}{202.23, 2.86} &  2526.14, 49.11 &38.57  \\ \cline{2-8} 
 \textbf{Insertion}& gSampler & \multicolumn{1}{c|}{4.32, \colorbox{pink}{1.29}} & \multicolumn{1}{c|}{4.61, \colorbox{pink}{2.05}} &12.12, \colorbox{pink}{8.24}& \multicolumn{1}{c|}{48.09, \colorbox{pink}{10.79}} & 695.43, 60.01 &9.42 
 \\ \cline{2-8}
 &FlowWalker&0.64, 0.70&2.92, 1.02&4.60, 2.99&18.57, 4.04&60108.03, 43.97&182.78\\
 \hline
 
& {\name} & \multicolumn{1}{c|}{\colorbox{green}{0.35}, 0.79} & \multicolumn{1}{c|}{\colorbox{green}{0.57}, 1.54} &\colorbox{green}{0.18}, 2.29& \multicolumn{1}{c|}{\colorbox{green}{5.92}, 6.93} & \colorbox{green}{66.51}, \colorbox{pink}{66.55} & - \\\cline{2-8}
\textbf{node2vec}& KnightKing & \multicolumn{1}{c|}{13.66, 0.33} & \multicolumn{1}{c|}{12.94, 0.44} &55.50, 1.00& \multicolumn{1}{c|}{176.38, 2.86} &  2550.16, 49.11& 87.64 \\ \cline{2-8} 
 \textbf{Deletion}& gSampler & \multicolumn{1}{c|}{3.81, \colorbox{pink}{1.06}} & \multicolumn{1}{c|}{4.14, \colorbox{pink}{1.59}} & 10.88, \colorbox{pink}{5.85}&\multicolumn{1}{c|}{43.42, \colorbox{pink}{7.39}} & 672.04, 60.01 & 19.21\\
 \cline{2-8}
 &FlowWalker&0.66, 0.71&2.94, 1.03&4.60, 2.99&18.62, 4.04&60007.51, 43.97&187.60\\
  \hline
 
& {\name} & \multicolumn{1}{c|}{\colorbox{green}{0.27}, 0.97} & \multicolumn{1}{c|}{\colorbox{green}{0.81}, 1.58} &\colorbox{green}{0.46}, 5.40& \multicolumn{1}{c|}{\colorbox{green}{6.14}, 8.13} &  \colorbox{green}{66.64}, 66.55 & - \\\cline{2-8}
\textbf{node2vec}& KnightKing & \multicolumn{1}{c|}{18.34, 0.33} & \multicolumn{1}{c|}{29.99, 0.43} &36.49, 1.00& \multicolumn{1}{c|}{109.78, 2.86}  & 2510.34, 49.22 & 47.97\\ \cline{2-8} 
 \textbf{Mixed}& gSampler & \multicolumn{1}{c|}{6.01, \colorbox{pink}{1.29}} & \multicolumn{1}{c|}{4.36, \colorbox{pink}{2.19}} &
 13.87, \colorbox{pink}{8.02}
 &\multicolumn{1}{c|}{44.03, \colorbox{pink}{12.07}} & 671.59, \colorbox{pink}{ 72.88} & 15.01 \\ \cline{2-8}
 &FlowWalker&0.66, 0.70&2.99, 1.02&4,64, 2.99&18.70, 4.04&59840.22, 43.97&183.45\\
 \hline
 
 & {\name} & \multicolumn{1}{c|}{\colorbox{green}{0.56}, 0.70} & \multicolumn{1}{c|}{\colorbox{green}{1.10}, 1.54} &\colorbox{green}{0.62}, 2.09 &\multicolumn{1}{c|}{\colorbox{green}{2.01}, 6.34} &\colorbox{green}{20.21}, 51.22 & - \\\cline{2-8}
\textbf{PPR}& KnightKing & \multicolumn{1}{c|}{2.67, 0.29} & \multicolumn{1}{c|}{4.52, 0.37} &11.86, 0.78&\multicolumn{1}{c|}{74.08, 2.54} &  1172.27, 46.14 & 24.57 \\ \cline{2-8} 
 \textbf{Insertion}& gSampler & \multicolumn{1}{c|}{1.67,\colorbox{pink}{ 1.43}} & \multicolumn{1}{c|}{2.42,\colorbox{pink}{ 2.35}} &
 3.86, \colorbox{pink}{9.74}
 &\multicolumn{1}{c|}{28.36, \colorbox{pink}{10.60}} & 519.07, \colorbox{pink}{52.76} & 10.24 \\ \cline{2-8}
 &FlowWalker&0.57, 0.70&1.93, 1.02&3.80, 2.99&17.25, 4.04&24230.61, 43.97&243.28\\
 \hline

& {\name} & \multicolumn{1}{c|}{\colorbox{green}{0.27}, 0.70} & \multicolumn{1}{c|}{\colorbox{green}{0.34}, 1.44} &\colorbox{green}{0.16}, 2.09 &\multicolumn{1}{c|}{\colorbox{green}{1.47}, 6.24} &  \colorbox{green}{19.69}, 51.12 & - \\\cline{2-8}
\textbf{PPR}& KnightKing & \multicolumn{1}{c|}{2.56, 0.29} & \multicolumn{1}{c|}{4.13, 0.36} &11.32, 0.78 &\multicolumn{1}{c|}{69.85, 2.54} & 1122.42, 46.14 & 39.38\\ \cline{2-8} 
 \textbf{Deletion}& gSampler & \multicolumn{1}{c|}{1.33, \colorbox{pink}{1.13}} & \multicolumn{1}{c|}{2.39, \colorbox{pink}{1.75}} & 3.56, \colorbox{pink}{6.74} &\multicolumn{1}{c|}{21.97, \colorbox{pink}{7.12}} &  476.82, \colorbox{pink}{72.54} & 14.67\\\cline{2-8}
 &FlowWalker&0.53, 0.70&1.94, 1.02&3.78, 2.99&17.04, 4.04&24181.06, 43.97& 254.19\\
 \hline

& {\name} & \multicolumn{1}{c|}{\colorbox{green}{0.26}, 0.87} & \multicolumn{1}{c|}{\colorbox{green}{0.66}, 1.38} &\colorbox{green}{0.44}, 4.62& \multicolumn{1}{c|}{\colorbox{green}{1.74}, 7.35} &  \colorbox{green}{19.96}, 51.22 & - \\ \cline{2-8}
\textbf{PPR}& KnightKing & \multicolumn{1}{c|}{2.58, 0.29} & \multicolumn{1}{c|}{4.25, 0.37} &11.08, 0.77& \multicolumn{1}{c|}{64.76, 2.53}  & 989.14, 46.21 &25.66\\ \cline{2-8} 
 \textbf{Mixed}& gSampler & \multicolumn{1}{c|}{3.39, \colorbox{pink}{1.44}} & \multicolumn{1}{c|}{2.58, \colorbox{pink}{2.53}} &
 4.52, \colorbox{pink}{9.45}
 &\multicolumn{1}{c|}{19.73, \colorbox{pink}{12.19}} & 559.58, \colorbox{pink}{53.97} & 13.32 \\
\cline{2-8}
 &FlowWalker&0.55, 0.70&1.97, 1.02&3.86, 2.99&17.23, 4.04&24243.25, 43.97&247.67\\ \hline
\end{tabular}
}
\end{table*}

\stitle{Datasets.} TABLE~\ref{dataset} presents the five real-world graph datasets, retrieved from Konect~\cite{konect} and SNAP~\cite{snapnets}, which we use for evaluation. 
We follow a three-step design to create dynamic updates: (i) we split the original graph dataset into two sets: A (original edges - $10\cdot BATCHSIZE$ edges) and B ($10\cdot BATCHSIZE$ edges) randomly. (ii) We randomly determine whether we want to delete or insert an edge. (iii) If we want to delete an edge, we will delete a randomly selected edge from set A. Otherwise, we randomly choose an edge from the set B and add that to set A. We perform this $10\cdot BATCHSIZE$ times to generate a sequence of $10\cdot BATCHSIZE$ updates. In this paper, we set $BATCHSIZE=100k$. We use the edges in set A as the edges to initialize the test. 

\ul{Dynamic updates.}
We generate three types of updating situations for each dataset in TABLE~\ref{dataset}: ``Insertion'', ``Deletion'', and ``Mixed'', which contain insertion only, deletion only, and mixed graph updates with an equal number of insertions and deletions, respectively. Except for Section~\ref{subsec:sota}, which uses all these three update situations, other sections only use the ``Mixed'' update situation (if not stated otherwise).

\ul{Bias.}
We generate the bias for most of the tests based on the degree of vertices, which naturally follow power law distribution (given all datasets are real-world graphs~\cite{chakrabarti2004r}). We also evaluate {\name} under floating-point bias and bias with different distributions in Section~\ref{subsec:microbenchmarks}.

\stitle{Evaluation Workflow.}
We perform (i) $BATCHSIZE$ number of updates and (ii) graph application computation. Since we generate 10 graph updates of $BATCHSIZE$, we repeat the aforementioned steps (i) and (ii) 10 rounds. We report the total time of these 10 rounds. 
By default, we conduct tests on GPU-based systems using a single GPU and on CPU-based systems using all the $16$ physical cores of a single machine.

\subsection{{\name} vs. the State-Of-The-Art (SOTA)}
\label{subsec:sota}

\stitle{Choice of SOTA.} We compare {\name} with three representative projects: KnightKing \cite{yang2019knightking}, gSampler \cite{gong2023gsampler} and FlowWalker \cite{mei2024flowwalker}. KnightKing is the first general-purpose CPU-based distributed graph random walk engine that employs the alias method for static biased sampling and the rejection method for dynamic biased sampling. Because Flashmob \cite{yang2021random}, a more recent work, does not support biased sampling, we choose KnightKing as the CPU-based SOTA. 
gSampler is one of the most recent GPU-based graph sampling systems, featuring matrix-based APIs for efficient execution. In addition, FlowWalker, another recent GPU-based sampling framework, is based on parallel reservoir sampling and serves as a strong baseline for comparison. 
Since the above systems only support static or streaming graphs, we develop their open-source code for our evaluation. Specifically, we reload or reconstruct the corresponding structure after each round of updates.
{Notice that we do not conduct comparative experiments with the dynamic graph processing systems, e.g., Wharf \cite{papadias2022space} or Hornet \cite{busato2018hornet} (see Section \ref{sub:dynamic}), as Bingo focuses on sampling, which is different from these projects. We will analyze the difference between {\name} and these projects in Section \ref{RW}.}

TABLE~\ref{tab:overall} provides a comprehensive analysis of {\name} vs the SOTA on various datasets and applications. For ease of viewing, we color-coded the \colorbox{green}{shortest runtime} and \colorbox{pink}{largest} \colorbox{pink}{memory consumption}.
Briefly, {\name} consistently outperforms KnightKing, 
gSampler, and FlowWalker by $24.46-112.28\times$, 
$8.74-25.66\times$ and $182.78-271.11\times$, respectively.

When it comes to runtime, we observe four different insights: (i) {\name} outperforms all SOTA across all graph applications on all graph datasets and updating cases. (ii) {\name} enjoys more speedup on larger graphs with bigger degrees. Using DeepWalk-Insertion as an example, our speedup climbs from 5.17$\times$ (Amazon) to 1279.9$\times$ (Twitter).
(iii) {\name} offers higher speedup on deletion than insertion, with mixed operations sitting in the middle. The reason lies in the fact that dynamic arrays may need to allocate memory immediately during insertion. In contrast, memory released during deletion can be managed offline without incurring immediate overhead in our custom memory pool. In short, deletion is more friendly to {\name} than insertion.

\begin{figure*}
    \subfloat[Overall.]
    {\includegraphics[width=0.205\linewidth]{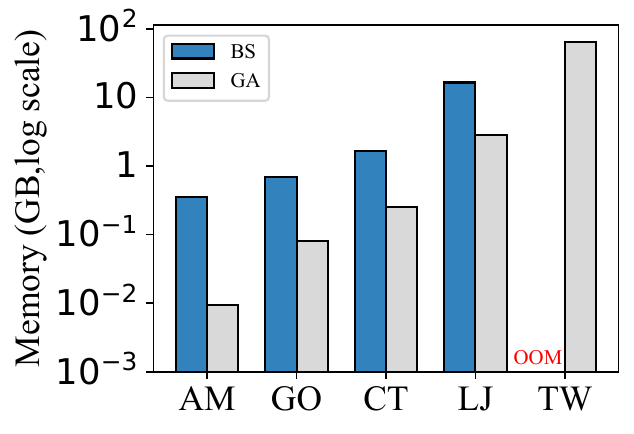}}
    \subfloat[Dense group.]
    {\includegraphics[width=0.19\linewidth]{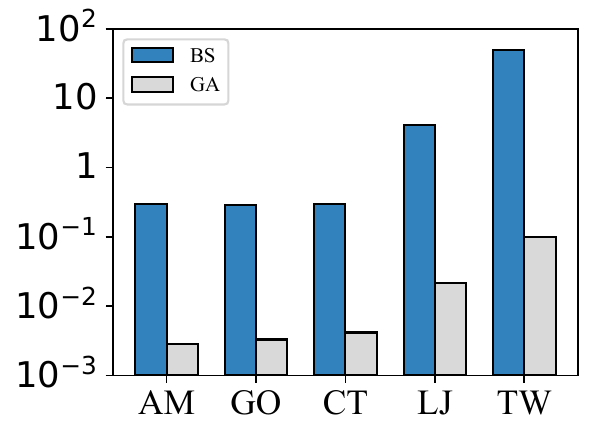}}
    \subfloat[One-element group.]
     {\includegraphics[width=0.19\linewidth]{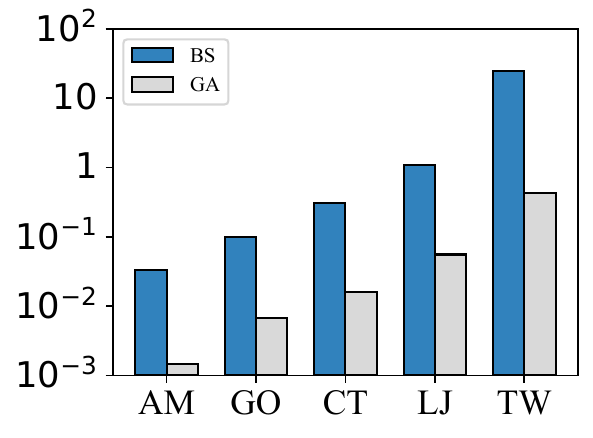}}
     \subfloat[Sparse group.]
     {\includegraphics[width=0.19\linewidth]{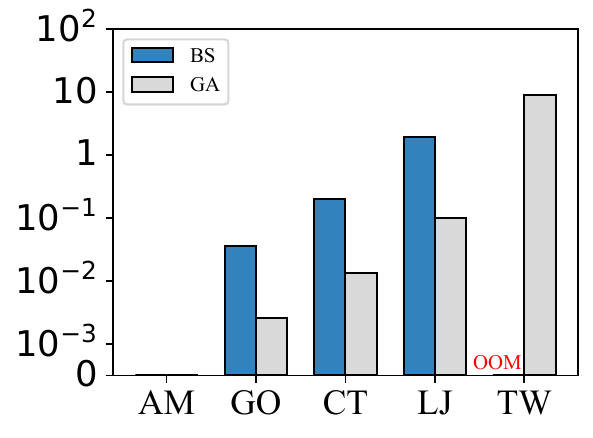}}
     \subfloat[Ratio of various groups.]
     {\includegraphics[width=0.19\linewidth]{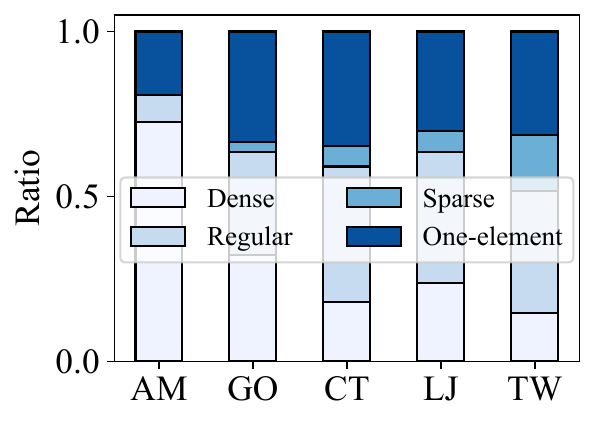}}
    \caption{Adaptive group representation impacts on memory consumption of {\name}, where BS, GA, respectively, represent ``BaSeline'' and ``Group Adaption optimization''. In BS, we use regular group format for all groups.}
    \label{Moptimization}  
\end{figure*}

Regarding memory consumption, we witness two key insights: (i) gSampler often consumes the most memory, followed by {\name}, then KnightKing and FlowWalker. The reason is that gSampler relies on matrix APIs, which will factor out a laundry list of memory costs. {\name} would consume more memory than KnightKing and FlowWalker because these two SOTAs rely on relatively simpler sampling space data structures or do not require such structures at all.
(ii) We notice a downside of {\name}. That is, {\name} consumes the most memory for three cases, i.e., DeepWalk-Deletion (Livejournal), node2vec-Insertion (Twitter), node2vec-Deletion (Twitter). The major reason is that these graphs contain more vertices with higher bias. That will render more regular groups in {\name}, which consumes more memory. Given {\name} outperforms the SOTA on larger graphs, we can slightly adjust $\alpha$ and $\beta$ in Equation~\ref{eq:group} when memory consumption is a concern. Further, {\name} is a distributed random walk system and can support higher radix factorizations, both of which can help mitigate the memory consumption problem faced by {\name}.

\subsection{Impacts of System Optimizations}

\textbf{Adaptive group representation.} 
Figure~\ref{Moptimization} illustrates the memory consumption savings brought by the adaptive group representation in {\name}. Figure~\ref{Moptimization}(a) shows that overall this optimization, on average, reduces the memory consumption from $14.6\times$ (GO) to $22.2\times$ (AM). Further, as shown in Figures~\ref{Moptimization}(b) - (d), Dense, One-element, and Sparse group representations, on average, reduce the memory by $323.67\times$, $21.51\times$, $6.41\times$, respectively. Overall, AM graph enjoys the most savings because this group contains the highest ratio of dense groups ($72.7$\%); see Figure~\ref{Moptimization}(e). It is worth noting that our GA optimization resolves the out-of-memory issue faced by the BS design for TW.

\begin{figure}[h]
    \subfloat[Insertion]{\includegraphics[width=0.34\linewidth]{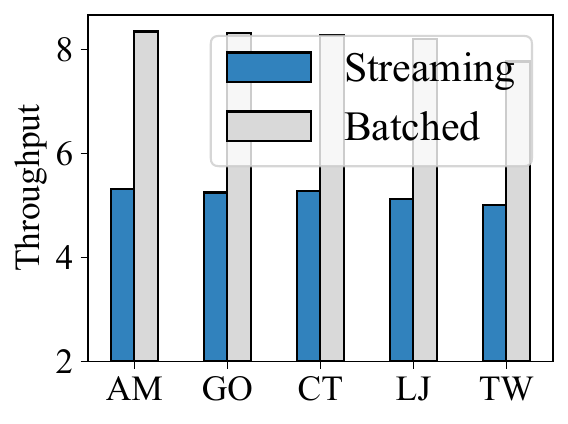}}
    \subfloat[Deletion]{\includegraphics[width=0.32\linewidth]{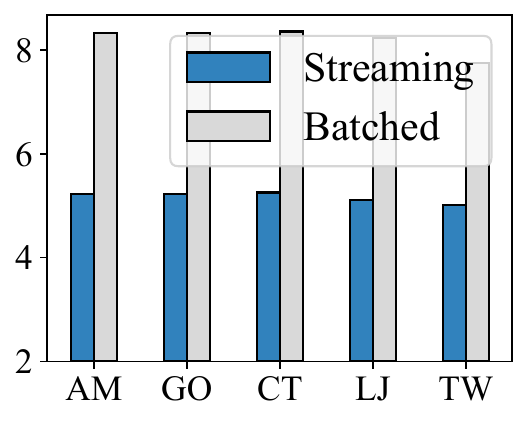}}
    \subfloat[Mixed]{\includegraphics[width=0.32\linewidth]{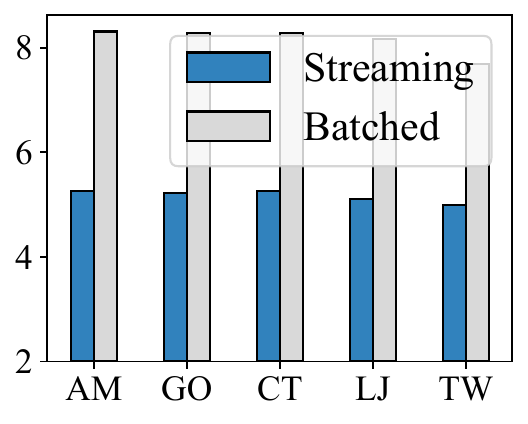}}
    \caption{Impacts of batched updates optimization.
    \vspace{-.2in}}
    \label{fig:stream}
\end{figure}

\stitle{Batched updates.}
Figure~\ref{fig:stream} exhibits the performance impacts of our batched graph update optimization.
In this study, we ingest ten $100K$ ``insertion'', ``deletion'', and ``mixed'' graph updates to {\name} in streamed vs. batched designs.
Overall, our batched update designs are, respectively, $1006.1$x, $1119.1$x, $992.5$x, faster than streaming update on ``insertion'', ``deletion'', and ``mixed'' update cases because (i) we can parallelize all the updates in batched updates, and (ii) we only perform one rebuild for all the updates. 
Further, ``delete'' enjoys the best speedup thanks to our 2-phase parallel delete-and-swap optimization. ``Mixed'' experiences the smallest speedup because one has to invoke both ``insert'' and ``delete'' kernels to perform the updates.

\begin{figure}[h]
     {\includegraphics[width=.8\linewidth]{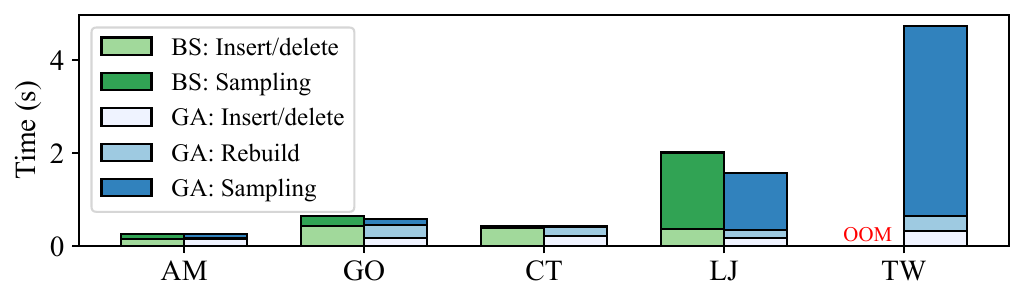}}
     \caption{Time consumption of BS vs GA, where BS and GA are defined in Figure~\ref{Moptimization}.}
          \label{fig:time_breakdown}
\end{figure}
\vspace{-.1in}

\begin{table}[h]
    \caption{Group conversion ratio in LJ graph.}
    \vspace{-.1in}
    \begin{center}
    {\scriptsize
        \begin{tabular}{|c|c|c|c|c|}
            \hline
            \textbf{Group Type} &Dense& Regular &Sparse  & One element  \\\hline
            Dense &---& 0.02\%    &0.01\%  &  0.47\%\\\hline
            Regular&0.01\%&---& <0.01\%&0.02\%\\\hline
            Sparse& <0.01\% & <0.01\% & --- & 0.14\%\\ \hline
            One element&0.05\%&0.03\%&0.01\%&---\\\hline
           
        \end{tabular}
   }
        \vspace{-.1in}
        \label{trans}
    \end{center}
\end{table}

\stitle{Impact on time cost.} 
Figure~\ref{fig:time_breakdown} presents the time consumption breakdown of with vs without GA optimization. Surprisingly, our GA optimization is, on average, 1.09$\times$ faster than BS in addition to the dramatic memory saving presented in Figure~\ref{Moptimization}. This speedup is offered from three aspects: (i) The sampling on GA optimization is $1.05$x - $1.61$x faster. (ii) Our insert/delete is faster than BS because updating one-element, sparse, and dense groups is faster, and (iii) as shown in TABLE~\ref{trans}, while rebuild introduces extra overheads compared to BS, the conversion from one group type to the other is very low for LJ graph. Particularly, the highest conversion rate is less than 0.47\%. This leads the time of GA insert/delete step together with the rebuild step to be merely 8\% slower than the BS insert/delete step for the worst case (AM).

\subsection{MicroBenchmarks}
\label{subsec:microbenchmarks}

This section uses DeepWalk with $100K$ mixed graph updates for microbenchmark evaluations.

\begin{figure}[h]
    {\includegraphics[width=0.45\linewidth]{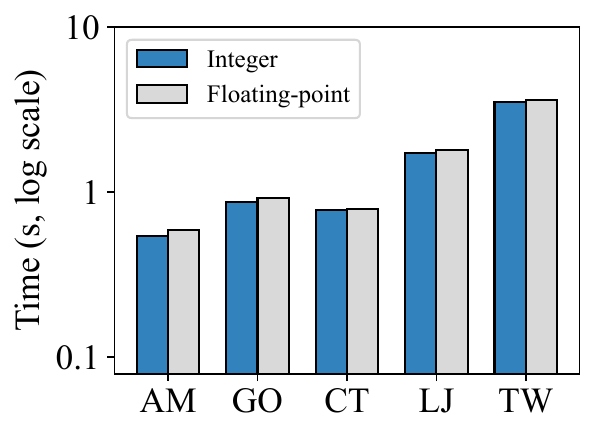}}
    {\includegraphics[width=0.45\linewidth]{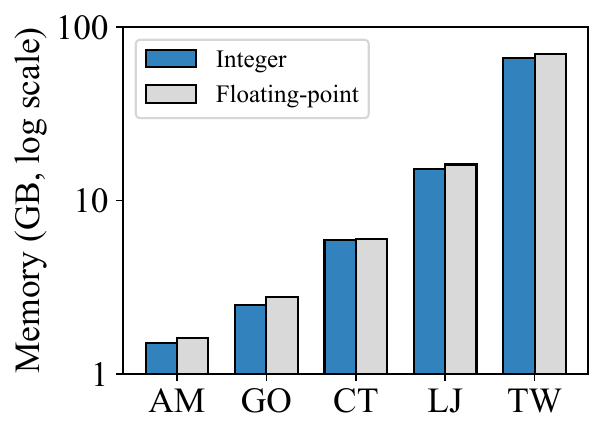}}
    \caption{{\name}: Integer bias vs. floating-point bias.}
    \label{fig:floating}
\end{figure}

\stitle{Floating-point vs. integer bias.} Figure~\ref{fig:floating} presents the performance of {\name} on the same dataset with different bias data types. For the fair comparison, the floating-point bias is the integer bias added with a random floating-point value between $0 - 1.00$.
Overall, the floating-point bias merely consumes, on average, 1.02$\times$ longer and 1.08$\times$ more memory than the integer counterpart, which is acceptable. 

\begin{figure}[h]
    \vspace{-.2in}
    \subfloat[Batch size.]{\includegraphics[width=0.33\linewidth]{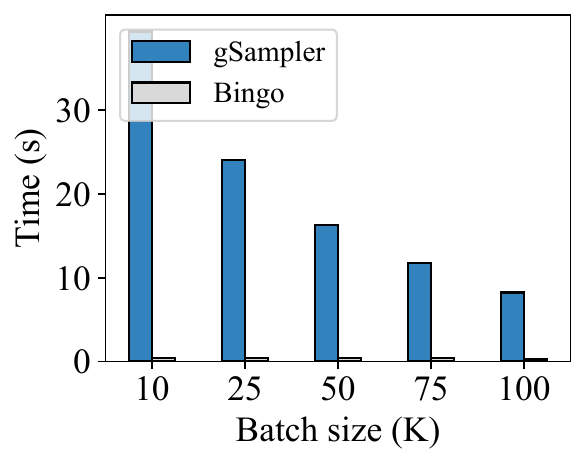}}
    \subfloat[Walk length.]{\includegraphics[width=0.30\linewidth]{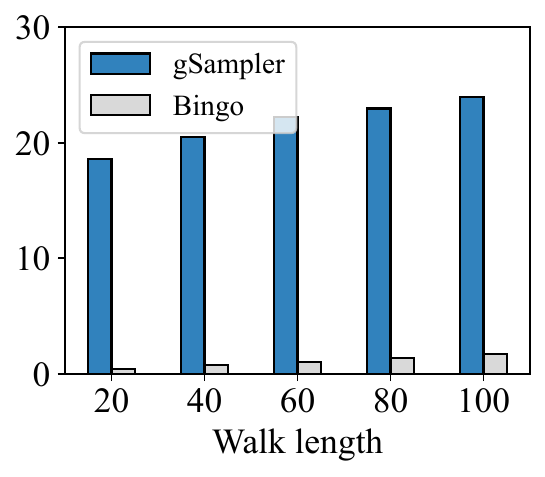}}
    \subfloat[Bias distribution.]{\includegraphics[width=0.36\linewidth]{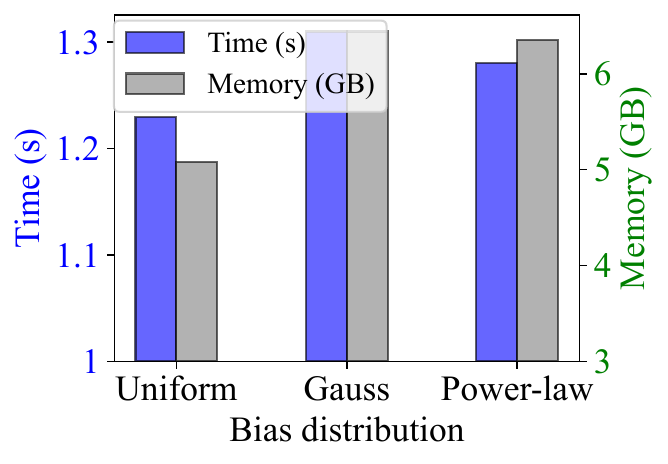}}
     \vspace{-.1in}
\caption{Varying evaluation configurations.}
\label{scal}
 \vspace{-.1in}
\end{figure}

\stitle{Varying evaluation configurations.}
Figure~\ref{scal}(a) shows the execution time of gSampler vs {\name} under different updating batch sizes for $1$ million updates on the LiveJournal dataset. We observe that the runtime of gSampler and {\name} decrease while the batch size increases. This is because larger batch size provides more parallel execution opportunities and helps reduce the rebuild time.

Figure~\ref{scal}(b) studies the performance of gSampler vs {\name} under different walk lengths. With the increase in walk lengths, indicating the increase in workloads, both gSampler and {\name} experience longer runtime. With the help of {\name} accelerated sampling, the gap between gSampler and {\name} also widens from $18.17$ s to $22.27$ s. 

Figure~\ref{scal}(c) evaluates the performance of {\name} under different bias distributions. {\name} on workload with a Uniform bias distribution consumes the least memory and time. This is because uniform bias distribution results in more dense groups and a lower rejection rate. In contrast, the other two distributions have relatively skewed bias distribution, requiring more memory and time to process.

\begin{figure}[h]
    \centering
    \subfloat[Updating time. {\name}\_I and {\name}\_D denote the time consumption of {\name} performing 1 million insertions, and 1 million deletions, respectively. FlowWalker\_R means FlowWalker reloads the graph after applying both the 1 million insertions and 1 million deletions.]{
        \includegraphics[width=0.8\linewidth]{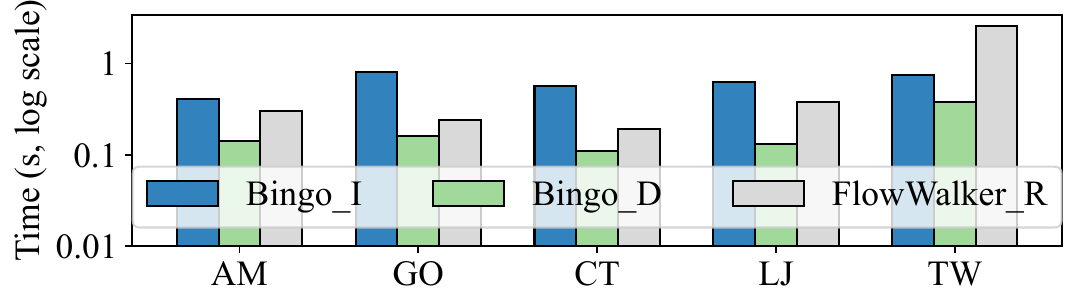}
    }\\[-0ex] 
    \subfloat[Sampling time.]{
        \includegraphics[width=0.8\linewidth]{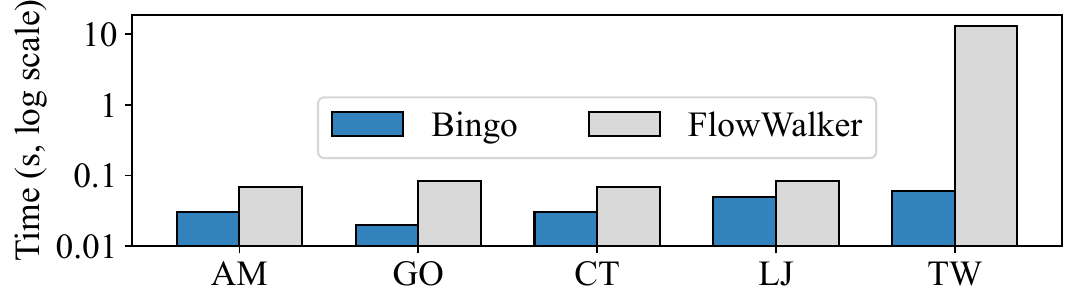}
    }
    \caption{Piecewise breakdown: {\name} vs. FlowWalker on updating and sampling time.}
    \label{fig:breakdown}
\end{figure}

\stitle{Piecewise breakdown.} Figure~\ref{fig:breakdown} studies the time consumption of 1 million insertions vs. 1 million deletions vs. 1 million sampling in {\name} and FlowWalker on various graph datasets. This study reveals the dominant factors affecting the overall performance. On average, {\name}'s insertion is $3.96 \times$ slower than deletion.  {\name}'s sampling is around $2$ orders of magnitude faster than deletion and insertion, thanks to $O(1)$ sampling time complexity.

Figure~\ref{fig:breakdown}(a) also compares the updating time of {\name} vs FlowWalker. FlowWalker achieves slightly faster ($2.35X$) updating performance than {\name} because it does not require maintaining auxiliary sampling structures. Instead, it simply reloads the new graph after updates. In contrast, {\name} must maintain additional sampling-related data structures, leading to marginally higher update times.

Figure~\ref{fig:breakdown}(b) further compares their sampling time. Both {\name} and FlowWalker achieve efficient sampling, but {\name} consistently outperforms FlowWalker across all datasets. Notably, in large-scale graphs like TW, {\name} maintains low sampling latency, while FlowWalker suffers a significant performance drop, with {\name} achieving a remarkable $218.7\times$ speedup. This gap stems from FlowWalker's reliance on reservoir sampling, which has an $\mathcal{O}(d)$ complexity. As the graph size and vertex degrees increase, this complexity becomes a bottleneck, significantly impacting sampling efficiency.

\section{Related Work}\label{RW}

\subsection{Random Walk and Sampling Systems}

\stitle{CPU-based systems.}
KnightKing \cite{yang2019knightking} pioneered the effort of building a dedicated graph engine for random walk applications. It introduces a walker-centric programming interface, divides the large biases into two smaller ones in pursuit of a lower rejection rate, and adopts 1-D graph partitioning to enable distributed random walk.
GraphWalker \cite{wang2020graphwalker}, on the other hand, builds a single machine-based random walk framework with a focus on I/O-efficient sub-graph loading strategy. ThunderRW \cite{sun2021thunderrw} notices that random walks severely under-utilize memory bandwidth. Thus, it breaks random walks into four steps and uses step interleaving to over-subscribe memory access requests. This design triggers aggressive software prefetching to improve memory bandwidth utilization. FlashMob \cite{yang2021random} simply re-arranges the memory operations of various sampling tasks to improve memory throughput during random walks.
Most recently, TEA \cite{huan2023tea} implements a hybrid sampling algorithm, which combines ITS and alias methods, for fast and memory-efficient sampling on a new type of graph, i.e., temporal graphs.

\stitle{Accelerator-based systems.}
C-SAW \cite{pandey2020c} leads the effort of building a general sampling framework on GPUs. This paper mainly optimizes the ITS method for first-order sampling. Inspired by KnightKing, NextDoor \cite{jangda2021accelerating} applies rejection sampling to GPUs and introduces transit parallelism for better load balancing and caching. Skywalker \cite{wang2021skywalker} implements parallel construction of alias tables on GPUs for improved load balancing, with Skywalker+~\cite{wangl2023optimizing} extending Skywalker to multiple GPUs. Recently, gSampler \cite{gong2023gsampler} introduced a set of expressive matrix-centric Application Programming Interface (APIs) to enhance the generality and efficiency of graph sampling for graph learning. We also noticed that LightRW~\cite{tan2023lightrw} has extended ThunderRW to FPGAs.

Whether on CPU, GPU, or FPGA platforms, none of these random walk systems explore random walk on dynamically changing graphs, which is the target of {\name}.

\vspace{-.1in}
\subsection{Dynamic Graph Processing Systems}
\label{sub:dynamic}
We identify a line of closely related interesting work which focuses on updating random walk results based on graph updates. These projects mainly aim to rapidly find the affected random walks that were computed previously for updating. Towards that end, Wharf \cite{papadias2022space} invents a compressed data structure that stores the random walk results with affordable memory and supports rapid results updates on dynamic graphs. FIRM \cite{hou2023personalized} presents an efficient indexing scheme that can trace and update the already calculated personalized PageRank results in constant time complexity when edge insertion and deletions are involved. Similar earlier efforts in this line are FORA~\cite{wang2019efficient} and SpeedPPR~\cite{wu2021unifying}.

Dynamic graph processing, which departs from traditional static graph processing systems\cite{zhu2016gemini,khorasani2014cusha,li2018regraph,zhu2019livegraph,maass2017mosaic,chen2019powerlyra,chen2017replication,venkataraman2013presto} or graph query systems\cite{wu2021fast,ding2019efficient,chen2023achieving}, is another related direction. First, on GPU platforms, STINGER \cite{ediger2012stinger} and cuSTINGER~\cite{green2016custinger} initialized the algorithm and data structure designs for frequent changes in dynamic graphs on GPUs. Hornet \cite{busato2018hornet} further improves cuSTINGER, while faimGraph \cite{winter2018faimgraph} designs a memory page management strategy for incremental updates on the GPU. Second, on the CPU platform, the streaming graph is a hot topic~\cite{jaiyeoba2019graphtinker,vora2017kickstarter,mariappan2019graphbolt,mariappan2021dzig}. Recent years have witnessed a surge of interest in maintaining ordered neighbors during graph updates for efficient streaming graph analytics~\cite{qi2024lsgraph}, examples are hash table \cite{awad2020dynamic}, Packed Memory Array~\cite{bender2007adaptive,sha2017technical}, Aspen~\cite{dhulipala2019low} and Pac-tree~\cite{dhulipala2022pac}. Terrace~\cite{pandey2021terrace} further adopts multiple data structures, including PMA and B-tree, for better performance. Most recently, LSGraph \cite{qi2024lsgraph} advocates support for both graph updates and graph computation analytics simultaneously with ordered neighbor updates.

{\name} differs from the aforementioned two lines of work as follows: (i) {\name} is orthogonal to identifying the already calculated random walks for the update. Particularly, once the calculated random walks are identified, instead of rebuilding the sampling space from scratch, {\name} can help them rapidly update the random walks. (ii) The second line of work (i.e., dynamic graph systems) provides a foundation for {\name}. In other words, these systems provide a platform but do not directly address instance generation for sampling or random walks. In particular, we adopt Hornet to support our dynamic data structures on GPUs.

\vspace{-.1in}
\subsection{Second-Order Random Walk}

Several projects have dubbed second-order random walks as dynamic graph random walks because these second-order algorithms (i.e., node2vec~\cite{grover2016node2vec}, Metapath~\cite{sun2013mining}, Second-order PageRank~\cite{wu2016remember}) require to change the transition probability of current vertex with respect to the history of a random walk. In short, these second-order algorithms introduce a dynamic bias for an unchanged static graph.

There mainly exist two lines of efforts in this direction, i.e., algorithm innovation and system designs:
(i) Regarding sampling algorithm design, KnightKing~\cite{yang2019knightking} introduces a novel two-step approach for second-order sampling: (1) using static sampling to select a vertex, and (2) using rejection sampling to involve the history of this random walk. On the contrary, FlowWalker~\cite{mei2024flowwalker} introduces two massively parallel sampling approaches based on Reservoir Sampling~\cite{vitter1985random} to quickly sample from the newly built sampling spaces.
(ii) For system design, GraSorw~\cite{li2022efficient} is a disk-based system for large graphs, using a triangular bi-block scheduling strategy to convert small random I/Os into large sequential I/Os. SOWalker~\cite{wu2023sowalker} optimizes I/O utilization by maximizing the benefit from block loading.

{\name} is orthogonal to these approaches because we work on algorithm and system designs of graphs with dynamic changing structures. We adopt Knighking's approach for handling second-order random walk applications, e.g., Node2vec.

\section{Conclusion}\label{C}
This paper takes the initiative to build a general random walk engine for dynamically changing graphs with two key principles: (i) this system should support both low-latency streaming updates and high-throughput batched updates. (ii) This system should achieve sampling speed and memory consumption comparable to the existing Monte Carlo sampling algorithms while supporting dynamic updates. Our system {\name} features three contributions:  
We first present a novel sampling algorithm that offers constant time sampling and fast updates. Furthermore, we introduce group adaptations for memory-efficient sampling space data structures. Finally, we introduce GPU-aware designs to support high-throughput batched graph updates. Our comprehensive evaluation demonstrates that {\name} outperforms the SOTA.

\begin{acks}
We thank the anonymous reviewers and shepherd Călin Iorgulescu for their helpful suggestions and feedback. 
This work was in part supported by the NSF CRII
Award No. 2331536, CAREER Award No. 2326141, and NSF Awards 2212370,
2319880, 2328948, 2411294, and 2417750.
\end{acks}

\bibliographystyle{ACM-Reference-Format}
\bibliography{references}


\begin{thebibliography}{77}


\ifx \showCODEN    \undefined \def \showCODEN     #1{\unskip}     \fi
\ifx \showDOI      \undefined \def \showDOI       #1{#1}\fi
\ifx \showISBNx    \undefined \def \showISBNx     #1{\unskip}     \fi
\ifx \showISBNxiii \undefined \def \showISBNxiii  #1{\unskip}     \fi
\ifx \showISSN     \undefined \def \showISSN      #1{\unskip}     \fi
\ifx \showLCCN     \undefined \def \showLCCN      #1{\unskip}     \fi
\ifx \shownote     \undefined \def \shownote      #1{#1}          \fi
\ifx \showarticletitle \undefined \def \showarticletitle #1{#1}   \fi
\ifx \showURL      \undefined \def \showURL       {\relax}        \fi
\providecommand\bibfield[2]{#2}
\providecommand\bibinfo[2]{#2}
\providecommand\natexlab[1]{#1}
\providecommand\showeprint[2][]{arXiv:#2}

\bibitem[Awad et~al\mbox{.}(2020)]%
        {awad2020dynamic}
\bibfield{author}{\bibinfo{person}{Muhammad~A Awad}, \bibinfo{person}{Saman Ashkiani}, \bibinfo{person}{Serban~D Porumbescu}, {and} \bibinfo{person}{John~D Owens}.} \bibinfo{year}{2020}\natexlab{}.
\newblock \showarticletitle{Dynamic graphs on the GPU}. In \bibinfo{booktitle}{\emph{IEEE International Parallel and Distributed Processing Symposium (IPDPS)}}. IEEE, \bibinfo{pages}{739--748}.
\newblock


\bibitem[Bajaj et~al\mbox{.}(2024)]%
        {bajaj2024graph}
\bibfield{author}{\bibinfo{person}{Saurabh Bajaj}, \bibinfo{person}{Hui Guan}, {and} \bibinfo{person}{Marco Serafini}.} \bibinfo{year}{2024}\natexlab{}.
\newblock \showarticletitle{Graph Neural Network Training Systems: A Performance Comparison of Full-Graph and Mini-Batch}.
\newblock \bibinfo{journal}{\emph{arXiv preprint arXiv:2406.00552}} (\bibinfo{year}{2024}).
\newblock


\bibitem[Bender and Hu(2007)]%
        {bender2007adaptive}
\bibfield{author}{\bibinfo{person}{Michael~A Bender} {and} \bibinfo{person}{Haodong Hu}.} \bibinfo{year}{2007}\natexlab{}.
\newblock \showarticletitle{An adaptive packed-memory array}.
\newblock \bibinfo{journal}{\emph{ACM Transactions on Database Systems (TODS)}} \bibinfo{volume}{32}, \bibinfo{number}{4} (\bibinfo{year}{2007}), \bibinfo{pages}{26--es}.
\newblock


\bibitem[Busato et~al\mbox{.}(2018)]%
        {busato2018hornet}
\bibfield{author}{\bibinfo{person}{Federico Busato}, \bibinfo{person}{Oded Green}, \bibinfo{person}{Nicola Bombieri}, {and} \bibinfo{person}{David~A Bader}.} \bibinfo{year}{2018}\natexlab{}.
\newblock \showarticletitle{Hornet: An efficient data structure for dynamic sparse graphs and matrices on gpus}. In \bibinfo{booktitle}{\emph{IEEE High Performance extreme Computing Conference (HPEC)}}. IEEE, \bibinfo{pages}{1--7}.
\newblock


\bibitem[Chakrabarti et~al\mbox{.}(2004)]%
        {chakrabarti2004r}
\bibfield{author}{\bibinfo{person}{Deepayan Chakrabarti}, \bibinfo{person}{Yiping Zhan}, {and} \bibinfo{person}{Christos Faloutsos}.} \bibinfo{year}{2004}\natexlab{}.
\newblock \showarticletitle{R-MAT: A recursive model for graph mining}. In \bibinfo{booktitle}{\emph{Proceedings of the SIAM International Conference on Data Mining}}. SIAM, \bibinfo{pages}{442--446}.
\newblock


\bibitem[Chen et~al\mbox{.}(2023a)]%
        {chen2023achieving}
\bibfield{author}{\bibinfo{person}{Hongtao Chen}, \bibinfo{person}{Mingxing Zhang}, \bibinfo{person}{Ke Yang}, \bibinfo{person}{Kang Chen}, \bibinfo{person}{Albert Zomaya}, \bibinfo{person}{Yongwei Wu}, {and} \bibinfo{person}{Xuehai Qian}.} \bibinfo{year}{2023}\natexlab{a}.
\newblock \showarticletitle{Achieving sub-second pairwise query over evolving graphs}. In \bibinfo{booktitle}{\emph{Proceedings of the 28th ACM International Conference on Architectural Support for Programming Languages and Operating Systems, Volume 2}}. \bibinfo{pages}{1--15}.
\newblock


\bibitem[Chen et~al\mbox{.}(2019)]%
        {chen2019powerlyra}
\bibfield{author}{\bibinfo{person}{Rong Chen}, \bibinfo{person}{Jiaxin Shi}, \bibinfo{person}{Yanzhe Chen}, \bibinfo{person}{Binyu Zang}, \bibinfo{person}{Haibing Guan}, {and} \bibinfo{person}{Haibo Chen}.} \bibinfo{year}{2019}\natexlab{}.
\newblock \showarticletitle{Powerlyra: Differentiated graph computation and partitioning on skewed graphs}.
\newblock \bibinfo{journal}{\emph{ACM Transactions on Parallel Computing (TOPC)}} \bibinfo{volume}{5}, \bibinfo{number}{3} (\bibinfo{year}{2019}), \bibinfo{pages}{1--39}.
\newblock


\bibitem[Chen et~al\mbox{.}(2017)]%
        {chen2017replication}
\bibfield{author}{\bibinfo{person}{Rong Chen}, \bibinfo{person}{Youyang Yao}, \bibinfo{person}{Peng Wang}, \bibinfo{person}{Kaiyuan Zhang}, \bibinfo{person}{Zhaoguo Wang}, \bibinfo{person}{Haibing Guan}, \bibinfo{person}{Binyu Zang}, {and} \bibinfo{person}{Haibo Chen}.} \bibinfo{year}{2017}\natexlab{}.
\newblock \showarticletitle{Replication-based fault-tolerance for large-scale graph processing}.
\newblock \bibinfo{journal}{\emph{IEEE Transactions on Parallel and Distributed Systems}} \bibinfo{volume}{29}, \bibinfo{number}{7} (\bibinfo{year}{2017}), \bibinfo{pages}{1621--1635}.
\newblock


\bibitem[Chen et~al\mbox{.}(2023b)]%
        {chen2023tango}
\bibfield{author}{\bibinfo{person}{Shiyang Chen}, \bibinfo{person}{Da Zheng}, \bibinfo{person}{Caiwen Ding}, \bibinfo{person}{Chengying Huan}, \bibinfo{person}{Yuede Ji}, {and} \bibinfo{person}{Hang Liu}.} \bibinfo{year}{2023}\natexlab{b}.
\newblock \showarticletitle{TANGO: Re-Thinking quantization for graph neural network training on GPUs}. In \bibinfo{booktitle}{\emph{Proceedings of the International Conference for High Performance Computing, Networking, Storage and Analysis}}. \bibinfo{pages}{1--14}.
\newblock


\bibitem[Cochez et~al\mbox{.}(2017)]%
        {cochez2017biased}
\bibfield{author}{\bibinfo{person}{Michael Cochez}, \bibinfo{person}{Petar Ristoski}, \bibinfo{person}{Simone~Paolo Ponzetto}, {and} \bibinfo{person}{Heiko Paulheim}.} \bibinfo{year}{2017}\natexlab{}.
\newblock \showarticletitle{Biased graph walks for RDF graph embeddings}. In \bibinfo{booktitle}{\emph{Proceedings of the 7th International Conference on Web Intelligence, Mining and Semantics}}. \bibinfo{pages}{1--12}.
\newblock


\bibitem[Dhulipala et~al\mbox{.}(2022)]%
        {dhulipala2022pac}
\bibfield{author}{\bibinfo{person}{Laxman Dhulipala}, \bibinfo{person}{Guy~E Blelloch}, \bibinfo{person}{Yan Gu}, {and} \bibinfo{person}{Yihan Sun}.} \bibinfo{year}{2022}\natexlab{}.
\newblock \showarticletitle{Pac-trees: Supporting parallel and compressed purely-functional collections}. In \bibinfo{booktitle}{\emph{Proceedings of the 43rd ACM SIGPLAN International Conference on Programming Language Design and Implementation}}. \bibinfo{pages}{108--121}.
\newblock


\bibitem[Dhulipala et~al\mbox{.}(2019)]%
        {dhulipala2019low}
\bibfield{author}{\bibinfo{person}{Laxman Dhulipala}, \bibinfo{person}{Guy~E Blelloch}, {and} \bibinfo{person}{Julian Shun}.} \bibinfo{year}{2019}\natexlab{}.
\newblock \showarticletitle{Low-latency graph streaming using compressed purely-functional trees}. In \bibinfo{booktitle}{\emph{Proceedings of the 40th ACM SIGPLAN conference on programming language design and implementation}}. \bibinfo{pages}{918--934}.
\newblock


\bibitem[Ding and Chen(2019)]%
        {ding2019efficient}
\bibfield{author}{\bibinfo{person}{Mengsu Ding} {and} \bibinfo{person}{Shimin Chen}.} \bibinfo{year}{2019}\natexlab{}.
\newblock \showarticletitle{Efficient partitioning and query processing of spatio-temporal graphs with trillion edges}. In \bibinfo{booktitle}{\emph{IEEE 35th International Conference on Data Engineering (ICDE)}}. IEEE, \bibinfo{pages}{1714--1717}.
\newblock


\bibitem[Ediger et~al\mbox{.}(2012)]%
        {ediger2012stinger}
\bibfield{author}{\bibinfo{person}{David Ediger}, \bibinfo{person}{Rob McColl}, \bibinfo{person}{Jason Riedy}, {and} \bibinfo{person}{David~A Bader}.} \bibinfo{year}{2012}\natexlab{}.
\newblock \showarticletitle{Stinger: High performance data structure for streaming graphs}. In \bibinfo{booktitle}{\emph{IEEE Conference on High Performance Extreme Computing}}. IEEE, \bibinfo{pages}{1--5}.
\newblock


\bibitem[Gong et~al\mbox{.}(2023)]%
        {gong2023gsampler}
\bibfield{author}{\bibinfo{person}{Ping Gong}, \bibinfo{person}{Renjie Liu}, \bibinfo{person}{Zunyao Mao}, \bibinfo{person}{Zhenkun Cai}, \bibinfo{person}{Xiao Yan}, \bibinfo{person}{Cheng Li}, \bibinfo{person}{Minjie Wang}, {and} \bibinfo{person}{Zhuozhao Li}.} \bibinfo{year}{2023}\natexlab{}.
\newblock \showarticletitle{gSampler: General and Efficient GPU-based Graph Sampling for Graph Learning}. In \bibinfo{booktitle}{\emph{Proceedings of the 29th Symposium on Operating Systems Principles}}. \bibinfo{pages}{562--578}.
\newblock


\bibitem[Green and Bader(2016)]%
        {green2016custinger}
\bibfield{author}{\bibinfo{person}{Oded Green} {and} \bibinfo{person}{David~A Bader}.} \bibinfo{year}{2016}\natexlab{}.
\newblock \showarticletitle{cuSTINGER: Supporting dynamic graph algorithms for GPUs}. In \bibinfo{booktitle}{\emph{IEEE High Performance Extreme Computing Conference (HPEC)}}. IEEE, \bibinfo{pages}{1--6}.
\newblock


\bibitem[Grover and Leskovec(2016)]%
        {grover2016node2vec}
\bibfield{author}{\bibinfo{person}{Aditya Grover} {and} \bibinfo{person}{Jure Leskovec}.} \bibinfo{year}{2016}\natexlab{}.
\newblock \showarticletitle{node2vec: Scalable feature learning for networks}. In \bibinfo{booktitle}{\emph{Proceedings of the 22nd ACM SIGKDD international conference on Knowledge discovery and data mining}}. \bibinfo{pages}{855--864}.
\newblock


\bibitem[Hamilton et~al\mbox{.}(2017)]%
        {hamilton2017inductive}
\bibfield{author}{\bibinfo{person}{Will Hamilton}, \bibinfo{person}{Zhitao Ying}, {and} \bibinfo{person}{Jure Leskovec}.} \bibinfo{year}{2017}\natexlab{}.
\newblock \showarticletitle{Inductive representation learning on large graphs}.
\newblock \bibinfo{journal}{\emph{Advances in neural information processing systems}}  \bibinfo{volume}{30} (\bibinfo{year}{2017}).
\newblock


\bibitem[Hang et~al\mbox{.}(2022)]%
        {hang2022outside}
\bibfield{author}{\bibinfo{person}{Jinquan Hang}, \bibinfo{person}{Zheng Dong}, \bibinfo{person}{Hongke Zhao}, \bibinfo{person}{Xin Song}, \bibinfo{person}{Peng Wang}, {and} \bibinfo{person}{Hengshu Zhu}.} \bibinfo{year}{2022}\natexlab{}.
\newblock \showarticletitle{Outside in: Market-aware heterogeneous graph neural network for employee turnover prediction}. In \bibinfo{booktitle}{\emph{Proceedings of the Fifteenth ACM International Conference on Web Search and Data Mining}}. \bibinfo{pages}{353--362}.
\newblock


\bibitem[Haveliwala et~al\mbox{.}(2003)]%
        {haveliwala2003analytical}
\bibfield{author}{\bibinfo{person}{Taher Haveliwala}, \bibinfo{person}{Sepandar Kamvar}, {and} \bibinfo{person}{Glen Jeh}.} \bibinfo{year}{2003}\natexlab{}.
\newblock \bibinfo{booktitle}{\emph{An analytical comparison of approaches to personalizing pagerank}}.
\newblock \bibinfo{type}{{T}echnical {R}eport}. \bibinfo{institution}{Stanford}.
\newblock


\bibitem[Hou et~al\mbox{.}(2023)]%
        {hou2023personalized}
\bibfield{author}{\bibinfo{person}{Guanhao Hou}, \bibinfo{person}{Qintian Guo}, \bibinfo{person}{Fangyuan Zhang}, \bibinfo{person}{Sibo Wang}, {and} \bibinfo{person}{Zhewei Wei}.} \bibinfo{year}{2023}\natexlab{}.
\newblock \showarticletitle{Personalized PageRank on Evolving Graphs with an Incremental Index-Update Scheme}.
\newblock \bibinfo{journal}{\emph{Proceedings of the ACM on Management of Data}} \bibinfo{volume}{1}, \bibinfo{number}{1} (\bibinfo{year}{2023}), \bibinfo{pages}{1--26}.
\newblock


\bibitem[Huan et~al\mbox{.}(2023)]%
        {huan2023tea}
\bibfield{author}{\bibinfo{person}{Chengying Huan}, \bibinfo{person}{Shuaiwen~Leon Song}, \bibinfo{person}{Santosh Pandey}, \bibinfo{person}{Hang Liu}, \bibinfo{person}{Yongchao Liu}, \bibinfo{person}{Baptiste Lepers}, \bibinfo{person}{Changhua He}, \bibinfo{person}{Kang Chen}, \bibinfo{person}{Jinlei Jiang}, {and} \bibinfo{person}{Yongwei Wu}.} \bibinfo{year}{2023}\natexlab{}.
\newblock \showarticletitle{Tea: a general-purpose temporal graph random walk engine}. In \bibinfo{booktitle}{\emph{Proceedings of the Eighteenth European Conference on Computer Systems}}. \bibinfo{pages}{182--198}.
\newblock


\bibitem[Jaiyeoba and Skadron(2019)]%
        {jaiyeoba2019graphtinker}
\bibfield{author}{\bibinfo{person}{Wole Jaiyeoba} {and} \bibinfo{person}{Kevin Skadron}.} \bibinfo{year}{2019}\natexlab{}.
\newblock \showarticletitle{Graphtinker: A high performance data structure for dynamic graph processing}. In \bibinfo{booktitle}{\emph{IEEE International Parallel and Distributed Processing Symposium (IPDPS)}}. IEEE, \bibinfo{pages}{1030--1041}.
\newblock


\bibitem[Jangda et~al\mbox{.}(2021)]%
        {jangda2021accelerating}
\bibfield{author}{\bibinfo{person}{Abhinav Jangda}, \bibinfo{person}{Sandeep Polisetty}, \bibinfo{person}{Arjun Guha}, {and} \bibinfo{person}{Marco Serafini}.} \bibinfo{year}{2021}\natexlab{}.
\newblock \showarticletitle{Accelerating graph sampling for graph machine learning using GPUs}. In \bibinfo{booktitle}{\emph{Proceedings of the Sixteenth European Conference on Computer Systems}}. \bibinfo{pages}{311--326}.
\newblock


\bibitem[Jeh and Widom(2002)]%
        {jeh2002simrank}
\bibfield{author}{\bibinfo{person}{Glen Jeh} {and} \bibinfo{person}{Jennifer Widom}.} \bibinfo{year}{2002}\natexlab{}.
\newblock \showarticletitle{Simrank: a measure of structural-context similarity}. In \bibinfo{booktitle}{\emph{Proceedings of the eighth ACM SIGKDD international conference on Knowledge discovery and data mining}}. \bibinfo{pages}{538--543}.
\newblock


\bibitem[Kazemi et~al\mbox{.}(2020)]%
        {kazemi2020representation}
\bibfield{author}{\bibinfo{person}{Seyed~Mehran Kazemi}, \bibinfo{person}{Rishab Goel}, \bibinfo{person}{Kshitij Jain}, \bibinfo{person}{Ivan Kobyzev}, \bibinfo{person}{Akshay Sethi}, \bibinfo{person}{Peter Forsyth}, {and} \bibinfo{person}{Pascal Poupart}.} \bibinfo{year}{2020}\natexlab{}.
\newblock \showarticletitle{Representation learning for dynamic graphs: A survey}.
\newblock \bibinfo{journal}{\emph{The Journal of Machine Learning Research}} \bibinfo{volume}{21}, \bibinfo{number}{1} (\bibinfo{year}{2020}), \bibinfo{pages}{2648--2720}.
\newblock


\bibitem[Khorasani et~al\mbox{.}(2014)]%
        {khorasani2014cusha}
\bibfield{author}{\bibinfo{person}{Farzad Khorasani}, \bibinfo{person}{Keval Vora}, \bibinfo{person}{Rajiv Gupta}, {and} \bibinfo{person}{Laxmi~N Bhuyan}.} \bibinfo{year}{2014}\natexlab{}.
\newblock \showarticletitle{CuSha: vertex-centric graph processing on GPUs}. In \bibinfo{booktitle}{\emph{Proceedings of the 23rd international symposium on High-performance parallel and distributed computing}}. \bibinfo{pages}{239--252}.
\newblock


\bibitem[Kunegis(2013)]%
        {konect}
\bibfield{author}{\bibinfo{person}{J\'{e}r\^{o}me Kunegis}.} \bibinfo{year}{2013}\natexlab{}.
\newblock \showarticletitle{{KONECT} -- {The} {Koblenz} {Network} {Collection}}. In \bibinfo{booktitle}{\emph{Proc. Int. Conf. on World Wide Web Companion}}. \bibinfo{pages}{1343--1350}.
\newblock
\urldef\tempurl%
\url{http://dl.acm.org/citation.cfm?id=2488173}
\showURL{%
\tempurl}


\bibitem[Leskovec and Krevl(2014)]%
        {snapnets}
\bibfield{author}{\bibinfo{person}{Jure Leskovec} {and} \bibinfo{person}{Andrej Krevl}.} \bibinfo{year}{2014}\natexlab{}.
\newblock \bibinfo{title}{{SNAP Datasets}: {Stanford} Large Network Dataset Collection}.
\newblock \bibinfo{howpublished}{\url{http://snap.stanford.edu/data}}.
\newblock


\bibitem[Lewis et~al\mbox{.}(2020)]%
        {lewis2020retrieval}
\bibfield{author}{\bibinfo{person}{Patrick Lewis}, \bibinfo{person}{Ethan Perez}, \bibinfo{person}{Aleksandra Piktus}, \bibinfo{person}{Fabio Petroni}, \bibinfo{person}{Vladimir Karpukhin}, \bibinfo{person}{Naman Goyal}, \bibinfo{person}{Heinrich K{\"u}ttler}, \bibinfo{person}{Mike Lewis}, \bibinfo{person}{Wen-tau Yih}, \bibinfo{person}{Tim Rockt{\"a}schel}, {et~al\mbox{.}}} \bibinfo{year}{2020}\natexlab{}.
\newblock \showarticletitle{Retrieval-augmented generation for knowledge-intensive nlp tasks}.
\newblock \bibinfo{journal}{\emph{Advances in Neural Information Processing Systems}}  \bibinfo{volume}{33} (\bibinfo{year}{2020}), \bibinfo{pages}{9459--9474}.
\newblock


\bibitem[Li et~al\mbox{.}(2022)]%
        {li2022efficient}
\bibfield{author}{\bibinfo{person}{Hongzheng Li}, \bibinfo{person}{Yingxia Shao}, \bibinfo{person}{Junping Du}, \bibinfo{person}{Bin Cui}, {and} \bibinfo{person}{Lei Chen}.} \bibinfo{year}{2022}\natexlab{}.
\newblock \showarticletitle{An I/O-efficient disk-based graph system for scalable second-order random walk of large graphs}.
\newblock \bibinfo{journal}{\emph{arXiv preprint arXiv:2203.16123}} (\bibinfo{year}{2022}).
\newblock


\bibitem[Li et~al\mbox{.}(2014)]%
        {li2014random}
\bibfield{author}{\bibinfo{person}{Rong-Hua Li}, \bibinfo{person}{Jeffrey~Xu Yu}, \bibinfo{person}{Xin Huang}, {and} \bibinfo{person}{Hong Cheng}.} \bibinfo{year}{2014}\natexlab{}.
\newblock \showarticletitle{Random-walk domination in large graphs}. In \bibinfo{booktitle}{\emph{IEEE 30th International Conference on Data Engineering}}. IEEE, \bibinfo{pages}{736--747}.
\newblock


\bibitem[Li et~al\mbox{.}(2018)]%
        {li2018regraph}
\bibfield{author}{\bibinfo{person}{Xue Li}, \bibinfo{person}{Mingxing Zhang}, \bibinfo{person}{Kang Chen}, {and} \bibinfo{person}{Yongwei Wu}.} \bibinfo{year}{2018}\natexlab{}.
\newblock \showarticletitle{Regraph: A graph processing framework that alternately shrinks and repartitions the graph}. In \bibinfo{booktitle}{\emph{Proceedings of the International Conference on Supercomputing}}. \bibinfo{pages}{172--183}.
\newblock


\bibitem[Luo et~al\mbox{.}(2023)]%
        {luo2023faf}
\bibfield{author}{\bibinfo{person}{Yice Luo}, \bibinfo{person}{Guannan Wang}, \bibinfo{person}{Yongchao Liu}, \bibinfo{person}{Jiaxin Yue}, \bibinfo{person}{Weihong Cheng}, {and} \bibinfo{person}{Binjie Fei}.} \bibinfo{year}{2023}\natexlab{}.
\newblock \showarticletitle{FAF: A Risk Detection Framework on Industry-Scale Graphs}. In \bibinfo{booktitle}{\emph{Proceedings of the 32nd ACM International Conference on Information and Knowledge Management}}. \bibinfo{pages}{4717--4723}.
\newblock


\bibitem[Ma et~al\mbox{.}(2023)]%
        {ma2023histgnn}
\bibfield{author}{\bibinfo{person}{Minbo Ma}, \bibinfo{person}{Peng Xie}, \bibinfo{person}{Fei Teng}, \bibinfo{person}{Bin Wang}, \bibinfo{person}{Shenggong Ji}, \bibinfo{person}{Junbo Zhang}, {and} \bibinfo{person}{Tianrui Li}.} \bibinfo{year}{2023}\natexlab{}.
\newblock \showarticletitle{Histgnn: Hierarchical spatio-temporal graph neural network for weather forecasting}.
\newblock \bibinfo{journal}{\emph{Information Sciences}}  \bibinfo{volume}{648} (\bibinfo{year}{2023}), \bibinfo{pages}{119580}.
\newblock


\bibitem[Maass et~al\mbox{.}(2017)]%
        {maass2017mosaic}
\bibfield{author}{\bibinfo{person}{Steffen Maass}, \bibinfo{person}{Changwoo Min}, \bibinfo{person}{Sanidhya Kashyap}, \bibinfo{person}{Woonhak Kang}, \bibinfo{person}{Mohan Kumar}, {and} \bibinfo{person}{Taesoo Kim}.} \bibinfo{year}{2017}\natexlab{}.
\newblock \showarticletitle{Mosaic: Processing a trillion-edge graph on a single machine}. In \bibinfo{booktitle}{\emph{Proceedings of the Twelfth European Conference on Computer Systems}}. \bibinfo{pages}{527--543}.
\newblock


\bibitem[Mariappan et~al\mbox{.}(2021)]%
        {mariappan2021dzig}
\bibfield{author}{\bibinfo{person}{Mugilan Mariappan}, \bibinfo{person}{Joanna Che}, {and} \bibinfo{person}{Keval Vora}.} \bibinfo{year}{2021}\natexlab{}.
\newblock \showarticletitle{DZiG: Sparsity-aware incremental processing of streaming graphs}. In \bibinfo{booktitle}{\emph{Proceedings of the sixteenth European conference on computer systems}}. \bibinfo{pages}{83--98}.
\newblock


\bibitem[Mariappan and Vora(2019)]%
        {mariappan2019graphbolt}
\bibfield{author}{\bibinfo{person}{Mugilan Mariappan} {and} \bibinfo{person}{Keval Vora}.} \bibinfo{year}{2019}\natexlab{}.
\newblock \showarticletitle{Graphbolt: Dependency-driven synchronous processing of streaming graphs}. In \bibinfo{booktitle}{\emph{Proceedings of the Fourteenth EuroSys Conference}}. \bibinfo{pages}{1--16}.
\newblock


\bibitem[Mei et~al\mbox{.}(2024)]%
        {mei2024flowwalker}
\bibfield{author}{\bibinfo{person}{Junyi Mei}, \bibinfo{person}{Shixuan Sun}, \bibinfo{person}{Chao Li}, \bibinfo{person}{Cheng Xu}, \bibinfo{person}{Cheng Chen}, \bibinfo{person}{Yibo Liu}, \bibinfo{person}{Jing Wang}, \bibinfo{person}{Cheng Zhao}, \bibinfo{person}{Xiaofeng Hou}, \bibinfo{person}{Minyi Guo}, {et~al\mbox{.}}} \bibinfo{year}{2024}\natexlab{}.
\newblock \showarticletitle{{FlowWalker: A Memory-efficient and High-performance GPU-based Dynamic Graph Random Walk Framework}}.
\newblock \bibinfo{journal}{\emph{arXiv preprint arXiv:2404.08364}} (\bibinfo{year}{2024}).
\newblock


\bibitem[Merrill and Grimshaw(2010)]%
        {merrill2010revisiting}
\bibfield{author}{\bibinfo{person}{Duane~G Merrill} {and} \bibinfo{person}{Andrew~S Grimshaw}.} \bibinfo{year}{2010}\natexlab{}.
\newblock \showarticletitle{Revisiting sorting for GPGPU stream architectures}. In \bibinfo{booktitle}{\emph{Proceedings of the 19th international conference on Parallel architectures and compilation techniques}}. \bibinfo{pages}{545--546}.
\newblock


\bibitem[Mikolov et~al\mbox{.}(2013)]%
        {mikolov2013efficient}
\bibfield{author}{\bibinfo{person}{Tomas Mikolov}, \bibinfo{person}{Kai Chen}, \bibinfo{person}{Greg Corrado}, {and} \bibinfo{person}{Jeffrey Dean}.} \bibinfo{year}{2013}\natexlab{}.
\newblock \showarticletitle{Efficient estimation of word representations in vector space}.
\newblock \bibinfo{journal}{\emph{arXiv preprint arXiv:1301.3781}} (\bibinfo{year}{2013}).
\newblock


\bibitem[Nikolentzos and Vazirgiannis(2020)]%
        {nikolentzos2020random}
\bibfield{author}{\bibinfo{person}{Giannis Nikolentzos} {and} \bibinfo{person}{Michalis Vazirgiannis}.} \bibinfo{year}{2020}\natexlab{}.
\newblock \showarticletitle{Random walk graph neural networks}.
\newblock \bibinfo{journal}{\emph{Advances in Neural Information Processing Systems}}  \bibinfo{volume}{33} (\bibinfo{year}{2020}), \bibinfo{pages}{16211--16222}.
\newblock


\bibitem[Pandey et~al\mbox{.}(2021)]%
        {pandey2021terrace}
\bibfield{author}{\bibinfo{person}{Prashant Pandey}, \bibinfo{person}{Brian Wheatman}, \bibinfo{person}{Helen Xu}, {and} \bibinfo{person}{Aydin Buluc}.} \bibinfo{year}{2021}\natexlab{}.
\newblock \showarticletitle{Terrace: A hierarchical graph container for skewed dynamic graphs}. In \bibinfo{booktitle}{\emph{Proceedings of the international conference on management of data}}. \bibinfo{pages}{1372--1385}.
\newblock


\bibitem[Pandey et~al\mbox{.}(2020)]%
        {pandey2020c}
\bibfield{author}{\bibinfo{person}{Santosh Pandey}, \bibinfo{person}{Lingda Li}, \bibinfo{person}{Adolfy Hoisie}, \bibinfo{person}{Xiaoye~S Li}, {and} \bibinfo{person}{Hang Liu}.} \bibinfo{year}{2020}\natexlab{}.
\newblock \showarticletitle{C-SAW: A framework for graph sampling and random walk on GPUs}. In \bibinfo{booktitle}{\emph{SC20: International Conference for High Performance Computing, Networking, Storage and Analysis}}. IEEE, \bibinfo{pages}{1--15}.
\newblock


\bibitem[Papadias et~al\mbox{.}(2022)]%
        {papadias2022space}
\bibfield{author}{\bibinfo{person}{Serafeim Papadias}, \bibinfo{person}{Zoi Kaoudi}, \bibinfo{person}{Jorge-Arnulfo Quian{\'e}-Ruiz}, {and} \bibinfo{person}{Volker Markl}.} \bibinfo{year}{2022}\natexlab{}.
\newblock \showarticletitle{Space-efficient random walks on streaming graphs}.
\newblock \bibinfo{journal}{\emph{Proceedings of the VLDB Endowment}} \bibinfo{volume}{16}, \bibinfo{number}{2} (\bibinfo{year}{2022}), \bibinfo{pages}{356--368}.
\newblock


\bibitem[Perozzi et~al\mbox{.}(2014)]%
        {perozzi2014deepwalk}
\bibfield{author}{\bibinfo{person}{Bryan Perozzi}, \bibinfo{person}{Rami Al-Rfou}, {and} \bibinfo{person}{Steven Skiena}.} \bibinfo{year}{2014}\natexlab{}.
\newblock \showarticletitle{Deepwalk: Online learning of social representations}. In \bibinfo{booktitle}{\emph{Proceedings of the 20th ACM SIGKDD international conference on Knowledge discovery and data mining}}. \bibinfo{pages}{701--710}.
\newblock


\bibitem[Polisetty et~al\mbox{.}(2023)]%
        {polisetty2023gsplit}
\bibfield{author}{\bibinfo{person}{Sandeep Polisetty}, \bibinfo{person}{Juelin Liu}, \bibinfo{person}{Kobi Falus}, \bibinfo{person}{Yi~Ren Fung}, \bibinfo{person}{Seung-Hwan Lim}, \bibinfo{person}{Hui Guan}, {and} \bibinfo{person}{Marco Serafini}.} \bibinfo{year}{2023}\natexlab{}.
\newblock \showarticletitle{GSplit: Scaling graph neural network training on large graphs via split-parallelism}.
\newblock \bibinfo{journal}{\emph{arXiv preprint arXiv:2303.13775}} (\bibinfo{year}{2023}).
\newblock


\bibitem[Qi et~al\mbox{.}(2024)]%
        {qi2024lsgraph}
\bibfield{author}{\bibinfo{person}{Hao Qi}, \bibinfo{person}{Yiyang Wu}, \bibinfo{person}{Ligang He}, \bibinfo{person}{Yu Zhang}, \bibinfo{person}{Kang Luo}, \bibinfo{person}{Minzhi Cai}, \bibinfo{person}{Hai Jin}, \bibinfo{person}{Zhan Zhang}, {and} \bibinfo{person}{Jin Zhao}.} \bibinfo{year}{2024}\natexlab{}.
\newblock \showarticletitle{LSGraph: a locality-centric high-performance streaming graph engine}. In \bibinfo{booktitle}{\emph{Proceedings of the Nineteenth European Conference on Computer Systems}}. \bibinfo{pages}{33--49}.
\newblock


\bibitem[Qiu et~al\mbox{.}(2018)]%
        {qiu2018real}
\bibfield{author}{\bibinfo{person}{Xiafei Qiu}, \bibinfo{person}{Wubin Cen}, \bibinfo{person}{Zhengping Qian}, \bibinfo{person}{You Peng}, \bibinfo{person}{Ying Zhang}, \bibinfo{person}{Xuemin Lin}, {and} \bibinfo{person}{Jingren Zhou}.} \bibinfo{year}{2018}\natexlab{}.
\newblock \showarticletitle{Real-time constrained cycle detection in large dynamic graphs}.
\newblock \bibinfo{journal}{\emph{Proceedings of the VLDB Endowment}} \bibinfo{volume}{11}, \bibinfo{number}{12} (\bibinfo{year}{2018}), \bibinfo{pages}{1876--1888}.
\newblock


\bibitem[Serafini and Guan(2021)]%
        {serafini2021scalable}
\bibfield{author}{\bibinfo{person}{Marco Serafini} {and} \bibinfo{person}{Hui Guan}.} \bibinfo{year}{2021}\natexlab{}.
\newblock \showarticletitle{Scalable graph neural network training: The case for sampling}.
\newblock \bibinfo{journal}{\emph{ACM SIGOPS Operating Systems Review}} \bibinfo{volume}{55}, \bibinfo{number}{1} (\bibinfo{year}{2021}), \bibinfo{pages}{68--76}.
\newblock


\bibitem[Sha et~al\mbox{.}(2017)]%
        {sha2017technical}
\bibfield{author}{\bibinfo{person}{Mo Sha}, \bibinfo{person}{Yuchen Li}, \bibinfo{person}{Bingsheng He}, {and} \bibinfo{person}{Kian-Lee Tan}.} \bibinfo{year}{2017}\natexlab{}.
\newblock \showarticletitle{Technical report: Accelerating dynamic graph analytics on gpus}.
\newblock \bibinfo{journal}{\emph{arXiv preprint arXiv:1709.05061}} (\bibinfo{year}{2017}).
\newblock


\bibitem[Sun et~al\mbox{.}(2021)]%
        {sun2021thunderrw}
\bibfield{author}{\bibinfo{person}{Shixuan Sun}, \bibinfo{person}{Yuhang Chen}, \bibinfo{person}{Shengliang Lu}, \bibinfo{person}{Bingsheng He}, {and} \bibinfo{person}{Yuchen Li}.} \bibinfo{year}{2021}\natexlab{}.
\newblock \showarticletitle{ThunderRW: an in-memory graph random walk engine}.
\newblock \bibinfo{journal}{\emph{Proceedings of the VLDB Endowment}} \bibinfo{volume}{14}, \bibinfo{number}{11} (\bibinfo{year}{2021}), \bibinfo{pages}{1992--2005}.
\newblock


\bibitem[Sun and Han(2013)]%
        {sun2013mining}
\bibfield{author}{\bibinfo{person}{Yizhou Sun} {and} \bibinfo{person}{Jiawei Han}.} \bibinfo{year}{2013}\natexlab{}.
\newblock \showarticletitle{Mining heterogeneous information networks: a structural analysis approach}.
\newblock \bibinfo{journal}{\emph{ACM SIGKDD explorations newsletter}} \bibinfo{volume}{14}, \bibinfo{number}{2} (\bibinfo{year}{2013}), \bibinfo{pages}{20--28}.
\newblock


\bibitem[Tan et~al\mbox{.}(2023a)]%
        {tan2023lightrw}
\bibfield{author}{\bibinfo{person}{Hongshi Tan}, \bibinfo{person}{Xinyu Chen}, \bibinfo{person}{Yao Chen}, \bibinfo{person}{Bingsheng He}, {and} \bibinfo{person}{Weng-Fai Wong}.} \bibinfo{year}{2023}\natexlab{a}.
\newblock \showarticletitle{LightRW: FPGA Accelerated Graph Dynamic Random Walks}.
\newblock \bibinfo{journal}{\emph{Proceedings of the ACM on Management of Data}} \bibinfo{volume}{1}, \bibinfo{number}{1} (\bibinfo{year}{2023}), \bibinfo{pages}{1--27}.
\newblock


\bibitem[Tan et~al\mbox{.}(2023b)]%
        {tan2023quiver}
\bibfield{author}{\bibinfo{person}{Zeyuan Tan}, \bibinfo{person}{Xiulong Yuan}, \bibinfo{person}{Congjie He}, \bibinfo{person}{Man-Kit Sit}, \bibinfo{person}{Guo Li}, \bibinfo{person}{Xiaoze Liu}, \bibinfo{person}{Baole Ai}, \bibinfo{person}{Kai Zeng}, \bibinfo{person}{Peter Pietzuch}, {and} \bibinfo{person}{Luo Mai}.} \bibinfo{year}{2023}\natexlab{b}.
\newblock \showarticletitle{Quiver: Supporting gpus for low-latency, high-throughput gnn serving with workload awareness}.
\newblock \bibinfo{journal}{\emph{arXiv preprint arXiv:2305.10863}} (\bibinfo{year}{2023}).
\newblock


\bibitem[Terdiman(2000)]%
        {radixSortRevisited}
\bibfield{author}{\bibinfo{person}{Pierre Terdiman}.} \bibinfo{year}{2000}\natexlab{}.
\newblock \bibinfo{title}{Radix Sort Revisited}.  (\bibinfo{date}{Apr.} \bibinfo{year}{2000}).
\newblock
\newblock
\shownote{\url{www.codercorner.com/RadixSortRevisited.htm}}.


\bibitem[Venkataraman et~al\mbox{.}(2013)]%
        {venkataraman2013presto}
\bibfield{author}{\bibinfo{person}{Shivaram Venkataraman}, \bibinfo{person}{Erik Bodzsar}, \bibinfo{person}{Indrajit Roy}, \bibinfo{person}{Alvin AuYoung}, {and} \bibinfo{person}{Robert~S Schreiber}.} \bibinfo{year}{2013}\natexlab{}.
\newblock \showarticletitle{Presto: distributed machine learning and graph processing with sparse matrices}. In \bibinfo{booktitle}{\emph{Proceedings of the 8th ACM European Conference on Computer Systems}}. \bibinfo{pages}{197--210}.
\newblock


\bibitem[Vitter(1985)]%
        {vitter1985random}
\bibfield{author}{\bibinfo{person}{Jeffrey~S Vitter}.} \bibinfo{year}{1985}\natexlab{}.
\newblock \showarticletitle{Random sampling with a reservoir}.
\newblock \bibinfo{journal}{\emph{ACM Transactions on Mathematical Software (TOMS)}} \bibinfo{volume}{11}, \bibinfo{number}{1} (\bibinfo{year}{1985}), \bibinfo{pages}{37--57}.
\newblock


\bibitem[Vora et~al\mbox{.}(2017)]%
        {vora2017kickstarter}
\bibfield{author}{\bibinfo{person}{Keval Vora}, \bibinfo{person}{Rajiv Gupta}, {and} \bibinfo{person}{Guoqing Xu}.} \bibinfo{year}{2017}\natexlab{}.
\newblock \showarticletitle{Kickstarter: Fast and accurate computations on streaming graphs via trimmed approximations}. In \bibinfo{booktitle}{\emph{Proceedings of the twenty-second international conference on architectural support for programming languages and operating systems}}. \bibinfo{pages}{237--251}.
\newblock


\bibitem[Wang et~al\mbox{.}(2021)]%
        {wang2021skywalker}
\bibfield{author}{\bibinfo{person}{Pengyu Wang}, \bibinfo{person}{Chao Li}, \bibinfo{person}{Jing Wang}, \bibinfo{person}{Taolei Wang}, \bibinfo{person}{Lu Zhang}, \bibinfo{person}{Jingwen Leng}, \bibinfo{person}{Quan Chen}, {and} \bibinfo{person}{Minyi Guo}.} \bibinfo{year}{2021}\natexlab{}.
\newblock \showarticletitle{Skywalker: Efficient alias-method-based graph sampling and random walk on GPUs}. In \bibinfo{booktitle}{\emph{30th International Conference on Parallel Architectures and Compilation Techniques (PACT)}}. IEEE, \bibinfo{pages}{304--317}.
\newblock


\bibitem[Wang et~al\mbox{.}(2023)]%
        {wangl2023optimizing}
\bibfield{author}{\bibinfo{person}{Pengyu Wang}, \bibinfo{person}{Cheng Xu}, \bibinfo{person}{Chao Li}, \bibinfo{person}{Jing Wang}, \bibinfo{person}{Taolei Wang}, \bibinfo{person}{Lu Zhang}, \bibinfo{person}{Xiaofeng Hou}, {and} \bibinfo{person}{Minyi Guo}.} \bibinfo{year}{2023}\natexlab{}.
\newblock \showarticletitle{Optimizing GPU-based Graph Sampling and Random Walk for Efficiency and Scalability}.
\newblock \bibinfo{journal}{\emph{IEEE Trans. Comput.}} (\bibinfo{year}{2023}).
\newblock


\bibitem[Wang et~al\mbox{.}(2020)]%
        {wang2020graphwalker}
\bibfield{author}{\bibinfo{person}{Rui Wang}, \bibinfo{person}{Yongkun Li}, \bibinfo{person}{Hong Xie}, \bibinfo{person}{Yinlong Xu}, {and} \bibinfo{person}{John~CS Lui}.} \bibinfo{year}{2020}\natexlab{}.
\newblock \showarticletitle{GraphWalker: An I/O-Efficient and Resource-Friendly Graph Analytic System for Fast and Scalable Random Walks}. In \bibinfo{booktitle}{\emph{USENIX Annual Technical Conference (USENIX ATC 20)}}. \bibinfo{pages}{559--571}.
\newblock


\bibitem[Wang et~al\mbox{.}(2019)]%
        {wang2019efficient}
\bibfield{author}{\bibinfo{person}{Sibo Wang}, \bibinfo{person}{Renchi Yang}, \bibinfo{person}{Runhui Wang}, \bibinfo{person}{Xiaokui Xiao}, \bibinfo{person}{Zhewei Wei}, \bibinfo{person}{Wenqing Lin}, \bibinfo{person}{Yin Yang}, {and} \bibinfo{person}{Nan Tang}.} \bibinfo{year}{2019}\natexlab{}.
\newblock \showarticletitle{Efficient algorithms for approximate single-source personalized pagerank queries}.
\newblock \bibinfo{journal}{\emph{ACM Transactions on Database Systems (TODS)}} \bibinfo{volume}{44}, \bibinfo{number}{4} (\bibinfo{year}{2019}), \bibinfo{pages}{1--37}.
\newblock


\bibitem[Wang et~al\mbox{.}(2011)]%
        {wang2011understanding}
\bibfield{author}{\bibinfo{person}{Tianyi Wang}, \bibinfo{person}{Yang Chen}, \bibinfo{person}{Zengbin Zhang}, \bibinfo{person}{Tianyin Xu}, \bibinfo{person}{Long Jin}, \bibinfo{person}{Pan Hui}, \bibinfo{person}{Beixing Deng}, {and} \bibinfo{person}{Xing Li}.} \bibinfo{year}{2011}\natexlab{}.
\newblock \showarticletitle{Understanding graph sampling algorithms for social network analysis}. In \bibinfo{booktitle}{\emph{31st international conference on distributed computing systems workshops}}. IEEE, \bibinfo{pages}{123--128}.
\newblock


\bibitem[Winter et~al\mbox{.}(2018)]%
        {winter2018faimgraph}
\bibfield{author}{\bibinfo{person}{Martin Winter}, \bibinfo{person}{Daniel Mlakar}, \bibinfo{person}{Rhaleb Zayer}, \bibinfo{person}{Hans-Peter Seidel}, {and} \bibinfo{person}{Markus Steinberger}.} \bibinfo{year}{2018}\natexlab{}.
\newblock \showarticletitle{faimGraph: High performance management of fully-dynamic graphs under tight memory constraints on the GPU}. In \bibinfo{booktitle}{\emph{SC18: International Conference for High Performance Computing, Networking, Storage and Analysis}}. IEEE, \bibinfo{pages}{754--766}.
\newblock


\bibitem[Wu et~al\mbox{.}(2021b)]%
        {wu2021unifying}
\bibfield{author}{\bibinfo{person}{Hao Wu}, \bibinfo{person}{Junhao Gan}, \bibinfo{person}{Zhewei Wei}, {and} \bibinfo{person}{Rui Zhang}.} \bibinfo{year}{2021}\natexlab{b}.
\newblock \showarticletitle{Unifying the global and local approaches: An efficient power iteration with forward push}. In \bibinfo{booktitle}{\emph{Proceedings of the International Conference on Management of Data}}. \bibinfo{pages}{1996--2008}.
\newblock


\bibitem[Wu et~al\mbox{.}(2021a)]%
        {wu2021fast}
\bibfield{author}{\bibinfo{person}{Jingqi Wu}, \bibinfo{person}{Rong Chen}, {and} \bibinfo{person}{Yubin Xia}.} \bibinfo{year}{2021}\natexlab{a}.
\newblock \showarticletitle{Fast and Accurate Optimizer for Query Processing over Knowledge Graphs}. In \bibinfo{booktitle}{\emph{Proceedings of the ACM Symposium on Cloud Computing}}. \bibinfo{pages}{503--517}.
\newblock


\bibitem[Wu et~al\mbox{.}(2016)]%
        {wu2016remember}
\bibfield{author}{\bibinfo{person}{Yubao Wu}, \bibinfo{person}{Yuchen Bian}, {and} \bibinfo{person}{Xiang Zhang}.} \bibinfo{year}{2016}\natexlab{}.
\newblock \showarticletitle{Remember where you came from: on the second-order random walk based proximity measures}.
\newblock \bibinfo{journal}{\emph{Proceedings of the VLDB Endowment}} \bibinfo{volume}{10}, \bibinfo{number}{1} (\bibinfo{year}{2016}), \bibinfo{pages}{13--24}.
\newblock


\bibitem[Wu et~al\mbox{.}(2023)]%
        {wu2023sowalker}
\bibfield{author}{\bibinfo{person}{Yutong Wu}, \bibinfo{person}{Zhan Shi}, \bibinfo{person}{Shicai Huang}, \bibinfo{person}{Zhipeng Tian}, \bibinfo{person}{Pengwei Zuo}, \bibinfo{person}{Peng Fang}, {and} \bibinfo{person}{Dan Feng}.} \bibinfo{year}{2023}\natexlab{}.
\newblock \showarticletitle{SOWalker: An I/O-Optimized Out-of-Core Graph Processing System for Second-Order Random Walks}. In \bibinfo{booktitle}{\emph{USENIX Annual Technical Conference (USENIX ATC 23)}}. \bibinfo{pages}{87--100}.
\newblock


\bibitem[Wu et~al\mbox{.}(2020)]%
        {wu2020comprehensive}
\bibfield{author}{\bibinfo{person}{Zonghan Wu}, \bibinfo{person}{Shirui Pan}, \bibinfo{person}{Fengwen Chen}, \bibinfo{person}{Guodong Long}, \bibinfo{person}{Chengqi Zhang}, {and} \bibinfo{person}{S~Yu Philip}.} \bibinfo{year}{2020}\natexlab{}.
\newblock \showarticletitle{A comprehensive survey on graph neural networks}.
\newblock \bibinfo{journal}{\emph{IEEE transactions on neural networks and learning systems}} \bibinfo{volume}{32}, \bibinfo{number}{1} (\bibinfo{year}{2020}), \bibinfo{pages}{4--24}.
\newblock


\bibitem[Xia et~al\mbox{.}(2019)]%
        {xia2019random}
\bibfield{author}{\bibinfo{person}{Feng Xia}, \bibinfo{person}{Jiaying Liu}, \bibinfo{person}{Hansong Nie}, \bibinfo{person}{Yonghao Fu}, \bibinfo{person}{Liangtian Wan}, {and} \bibinfo{person}{Xiangjie Kong}.} \bibinfo{year}{2019}\natexlab{}.
\newblock \showarticletitle{Random walks: A review of algorithms and applications}.
\newblock \bibinfo{journal}{\emph{IEEE Transactions on Emerging Topics in Computational Intelligence}} \bibinfo{volume}{4}, \bibinfo{number}{2} (\bibinfo{year}{2019}), \bibinfo{pages}{95--107}.
\newblock


\bibitem[Yang et~al\mbox{.}(2021b)]%
        {yang2021random}
\bibfield{author}{\bibinfo{person}{Ke Yang}, \bibinfo{person}{Xiaosong Ma}, \bibinfo{person}{Saravanan Thirumuruganathan}, \bibinfo{person}{Kang Chen}, {and} \bibinfo{person}{Yongwei Wu}.} \bibinfo{year}{2021}\natexlab{b}.
\newblock \showarticletitle{Random walks on huge graphs at cache efficiency}. In \bibinfo{booktitle}{\emph{Proceedings of the ACM SIGOPS 28th Symposium on Operating Systems Principles}}. \bibinfo{pages}{311--326}.
\newblock


\bibitem[Yang et~al\mbox{.}(2019)]%
        {yang2019knightking}
\bibfield{author}{\bibinfo{person}{Ke Yang}, \bibinfo{person}{MingXing Zhang}, \bibinfo{person}{Kang Chen}, \bibinfo{person}{Xiaosong Ma}, \bibinfo{person}{Yang Bai}, {and} \bibinfo{person}{Yong Jiang}.} \bibinfo{year}{2019}\natexlab{}.
\newblock \showarticletitle{Knightking: a fast distributed graph random walk engine}. In \bibinfo{booktitle}{\emph{Proceedings of the 27th ACM symposium on operating systems principles}}. \bibinfo{pages}{524--537}.
\newblock


\bibitem[Yang et~al\mbox{.}(2021a)]%
        {yang2021consisrec}
\bibfield{author}{\bibinfo{person}{Liangwei Yang}, \bibinfo{person}{Zhiwei Liu}, \bibinfo{person}{Yingtong Dou}, \bibinfo{person}{Jing Ma}, {and} \bibinfo{person}{Philip~S Yu}.} \bibinfo{year}{2021}\natexlab{a}.
\newblock \showarticletitle{Consisrec: Enhancing gnn for social recommendation via consistent neighbor aggregation}. In \bibinfo{booktitle}{\emph{Proceedings of the 44th international ACM SIGIR conference on Research and development in information retrieval}}. \bibinfo{pages}{2141--2145}.
\newblock


\bibitem[Zhou et~al\mbox{.}(2018)]%
        {zhou2018efficient}
\bibfield{author}{\bibinfo{person}{Dongyan Zhou}, \bibinfo{person}{Songjie Niu}, {and} \bibinfo{person}{Shimin Chen}.} \bibinfo{year}{2018}\natexlab{}.
\newblock \showarticletitle{Efficient graph computation for Node2Vec}.
\newblock \bibinfo{journal}{\emph{arXiv preprint arXiv:1805.00280}} (\bibinfo{year}{2018}).
\newblock


\bibitem[Zhu et~al\mbox{.}(2016)]%
        {zhu2016gemini}
\bibfield{author}{\bibinfo{person}{Xiaowei Zhu}, \bibinfo{person}{Wenguang Chen}, \bibinfo{person}{Weimin Zheng}, {and} \bibinfo{person}{Xiaosong Ma}.} \bibinfo{year}{2016}\natexlab{}.
\newblock \showarticletitle{Gemini: A $\{$Computation-Centric$\}$ distributed graph processing system}. In \bibinfo{booktitle}{\emph{12th USENIX Symposium on Operating Systems Design and Implementation (OSDI 16)}}. \bibinfo{pages}{301--316}.
\newblock


\bibitem[Zhu et~al\mbox{.}(2019)]%
        {zhu2019livegraph}
\bibfield{author}{\bibinfo{person}{Xiaowei Zhu}, \bibinfo{person}{Guanyu Feng}, \bibinfo{person}{Marco Serafini}, \bibinfo{person}{Xiaosong Ma}, \bibinfo{person}{Jiping Yu}, \bibinfo{person}{Lei Xie}, \bibinfo{person}{Ashraf Aboulnaga}, {and} \bibinfo{person}{Wenguang Chen}.} \bibinfo{year}{2019}\natexlab{}.
\newblock \showarticletitle{LiveGraph: A transactional graph storage system with purely sequential adjacency list scans}.
\newblock \bibinfo{journal}{\emph{arXiv preprint arXiv:1910.05773}} (\bibinfo{year}{2019}).
\newblock


\end{thebibliography}

\clearpage

\section{Supplement Materials}

\subsection{Dynamic Graph Support and Partitioning}

\stitle{Dynamic data structure support.} {\name}'s design is compatible with most dynamic graph processing systems. To take advantage of incremental updates, {\name} needs dynamic data structures for our sampling system. {\name} maintains an intra-group neighbor index list and inverted index, in addition to the graph adjacency list. The good news is both the intra-group neighbor index list and the inverted index are similar to the adjacency list. 

{\name} adopts the dynamic array design from Hornet~\cite{busato2018hornet} to support our design.
Particularly, there mainly exist four different solutions for maintaining this structure on dynamic graph systems: block linked list~\cite{green2016custinger,winter2018faimgraph}, dynamic array~\cite{busato2018hornet}, hash table~\cite{awad2020dynamic}, and hybrid data structure~\cite{pandey2021terrace}. Querying information is a frequent and significant operation in both update and sampling of {\name}. Because the hash table design introduces random memory indices with extra hash calculation, block-linked and hybrid data structures have hierarchical structures that need extra global memory indirection, we choose dynamic array design from Hornet. Of note, we also maintain memory pools for dynamic arrays to reduce the cost of memory allocation.

\stitle{Graph partition.} Our system can be scaled to multiple GPUs to combat potentially huge workloads. Similar to Knightking~\cite{yang2019knightking}, we adopt 1-D graph partitioning to distribute the graph and associated sampling data structure across GPUs. 
Communication and data transfer are unavoidable in multi-GPU {\name}. Instead of transferring the sampling data structure, we choose to transfer walkers in the procedure of random walking, which has the following advantage: The cost of transferring the sampling data structure might be larger than recalculating it while transferring walkers has the light burden of communication. Further, we leverage the fast peer-to-peer GPU communication to reduce the cost.

\subsection{{\name} with Arbitrary Radix Bases}\label{sec:disc}

\begin{figure}[h]
    \includegraphics[width=1\linewidth]{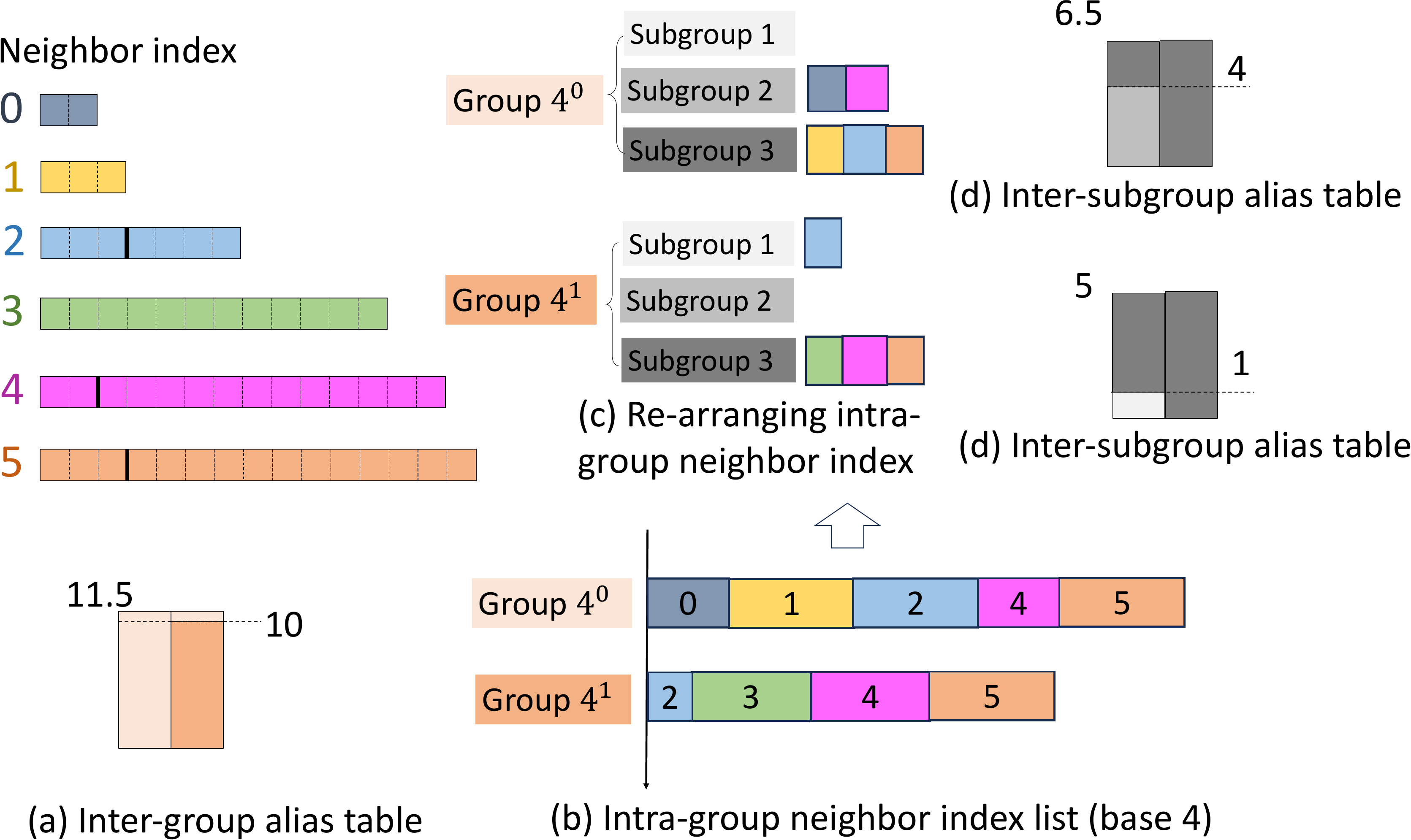}
    \caption{{\name} with radix base as 4.
    \vspace{-.2in}
    }
    \label{base4}
\end{figure}

\noindent

Figure~\ref{base4} exhibits {\name} with radix base as $4$. In this context, with another hierarchy or re-arrangement, we can still achieve near-constant time sampling. Following our original design, one will derive the (a) inter-group alias table and (b) intra-group neighbor index list. However, with base=4, the neighbors in each group of $4^i$ might not have the same bias. Using group $4^0$ as an example, neighbors 0 and 1, respectively, hold biases of 2 and 3 (their lengths indicate their bias values). 

Figure~\ref{base4}(c) and~\ref{base4}(d) offer a solution for this problem by further grouping various values in each group into subgroups. Using group $4^0$ of Figure~\ref{base4}(c) as an example, it has two neighbors in subgroup 2 and three in subgroup 3. Further, we can build an inter-subgroup alias table in Figure~\ref{base4}(d). During sampling, we can first identify the group. Subsequently, we can use the inter-subgroup alias table to identify the subgroup. Finally, one can still enjoy unbiased sampling in each subgroup to select the neighbor of interest. 

According to Table~\ref{timecmp}, larger radix bases could potentially help reduce the number of groups $K$ thus the complexity of insertion, deletion, and memory consumption. However, building a higher hierarchy of nested dynamic data structures could be challenging on GPUs. We thus did not evaluate {\name} on higher radix bases. This design might shed some light for CPU-based explorations.

\end{document}